\documentclass[onefignum,onetabnum]{siamart190516}

\usepackage{lipsum}
\usepackage{amsfonts}
\usepackage{amsmath}
\usepackage{amssymb}
\usepackage{bbm}
\usepackage[mathscr]{euscript}
\usepackage{latexsym}

\usepackage{algorithmic}
\usepackage{enumitem}

\usepackage{graphicx}
\usepackage{epstopdf}
\usepackage{tikz}
\usetikzlibrary{cd}
\usepackage{onimage}
\usepackage[caption=false]{subfig}
\ifpdf
  \DeclareGraphicsExtensions{.eps,.pdf,.png,.jpg}
\else
  \DeclareGraphicsExtensions{.eps}
\fi

\newsiamremark{remark}{Remark}
\newsiamremark{example}{Example}
\newsiamremark{hypothesis}{Hypothesis}
\crefname{hypothesis}{Hypothesis}{Hypotheses}
\newsiamthm{claim}{Claim}
\newsiamthm{assumption}{Assumption}

\headers{Covariance Balancing Reduction}{S. E. Otto, A. Padovan, and C. W. Rowley}

\title{
Model Reduction for Nonlinear Systems by Balanced Truncation of State and Gradient Covariance
\thanks{Compiled \today.
\funding{This research was supported by the Army Research Office under grant
  number W911NF-17-1-0512 and the Air Force Office of Scientific Research under
  grant number FA9550-19-1-0005.  S.E.O. was supported by the National Science Foundation Graduate Research Fellowship Program under Grant No. DGE-2039656.}}}

\author{
Samuel E. Otto\thanks{Princeton University, Princeton, NJ, Dept. of Mechanical and Aerospace Engineering (S.E.O.: \email{sotto@princeton.edu}, A.P.: \email{apadovan@princeton.edu}, C.W.R.: \email{cwrowley@princeton.edu})}
\and Alberto Padovan\footnotemark[2]
\and Clarence W. Rowley\footnotemark[2]
}

\DeclareMathOperator{\avg}{avg}

\DeclareMathOperator{\E}{\mathbb{E}}

\usepackage{amsopn}

\DeclareMathOperator{\Tr}{Tr}
\DeclareMathOperator{\rank}{rank}
\DeclareMathOperator{\diag}{diag}

\DeclareMathOperator{\Range}{Range}
\DeclareMathOperator{\Null}{Null}

\DeclareMathOperator{\col}{col}
\DeclareMathOperator{\row}{row}
\DeclareMathOperator{\grad}{\nabla}

\DeclareMathOperator{\D}{\mathrm{D}}

\newcommand{\scolon}{\, {:} \,}

\newcommand{\R}{\mathbb{R}}
\newcommand{\calG}{\mathcal{G}}
\newcommand{\calH}{\mathcal{H}}
\newcommand{\calU}{\mathcal{U}}
\newcommand{\calV}{\mathcal{V}}
\newcommand{\calW}{\mathcal{W}}
\newcommand{\calX}{\mathcal{X}}

\makeatletter
\renewcommand*\env@matrix[1][*\c@MaxMatrixCols c]{%
  \hskip -\arraycolsep
  \let\@ifnextchar\new@ifnextchar
  \array{#1}}
\makeatother

\usepackage{tikz,pgfplots}

\pgfplotsset{compat=1.17}

\begin{document}

\maketitle

\begin{abstract}
Data-driven reduced-order models often fail to make accurate forecasts of high-dimensional nonlinear dynamical systems that are sensitive along coordinates with low-variance because such coordinates are often truncated, e.g., by proper orthogonal decomposition, kernel principal component analysis, and autoencoders.
Such systems are encountered frequently in shear-dominated fluid flows where non-normality plays a significant role in the growth of disturbances.
In order to address these issues, we employ ideas from active subspaces to find low-dimensional systems of coordinates for model reduction that balance adjoint-based information about the system's sensitivity with the variance of states along trajectories.
The resulting method, which we refer to as covariance balancing reduction using adjoint snapshots (CoBRAS), is analogous to balanced truncation with state and adjoint-based gradient covariance matrices replacing the system Gramians and obeying the same key transformation laws.
Here, the extracted coordinates are associated with an oblique projection that can be used to construct Petrov-Galerkin reduced-order models.
We provide an efficient snapshot-based computational method analogous to balanced proper orthogonal decomposition.
This also leads to the observation that the reduced coordinates can be computed relying on inner products of state and gradient samples alone, allowing
us to find rich nonlinear coordinates by replacing the inner product with a kernel function.
In these coordinates, reduced-order models can be learned using regression.
We demonstrate these techniques and compare to a variety of other methods on a simple, yet challenging three-dimensional system and a nonlinear axisymmetric jet flow simulation with $10^5$ state variables.
\end{abstract}




\section{Introduction}
\label{sec:introduction}

Models that describe the time-evolution of physical systems such as fluid flows are key to forecasting, tracking, and controlling their behavior.
While the governing equations for these systems, e.g., the Navier-Stokes equations, can often be simulated, these simulations are too costly to be used in real-time applications.
Instead, a simplified reduced-order model (ROM) can be used to describe the most important aspects of the original system's behavior.
Methods for constructing ROMs from the original full-order model (FOM) entail finding a suitable projection that reduces the dimension of the state space without neglecting important information.
Many different approaches are possible depending on what state information is deemed important.
For reviews of modern methods see \cite{Ghadami2022data, Rowley2017model, Benner2015survey}.

The simplest and by far the most common approach is to say that a ``feature'' or state-space coordinate contains important information about the system if it explains a large amount of variance among the states generated along typical trajectories.
The optimal linear features in this respect are given by principal component analysis (PCA) also known in some communities as proper orthogonal decomposition (POD).
The first application of POD to model reduction of fluid flows was by Lumley \cite{Lumley1967structure}, with later work by Sirovich \cite{Sirovich1987turbulence} and Holmes et al. \cite{Holmes2012turbulence}.
As noted by Sirovich \cite{Sirovich1987turbulence}, the PCA/POD basis vectors or ``modes'' can be computed using the ``method of snapshots'' without assembling the state covariance matrix.
The resulting coordinates of a state vector projected onto the leading modes can be found using only inner products between the given state vector and the original data.
Sch{\"o}lkopf \cite{Scholkopf1998nonlinear} used this observation to develop kernel PCA (KPCA), which extracts rich nonlinear features by applying PCA/POD after lifting data into a higher-dimensional (possibly infinite-dimensional) reproducing kernel Hilbert space (RKHS).
Nonlinear features that are optimal for reconstructing the system's states can also be found using autoencoder neural networks and employed for model reduction as described by Lee and Carlberg \cite{Lee2020model}.

Non-normality in shear-dominated fluid flows causes the evolution of these systems to be heavily influenced by certain low-energy (low-variance) disturbances \cite{Trefethen-93, Schmid2001stability}.
Therefore, the features (coordinates) we choose to construct a ROM must account for both the variance of the system's states as well as the sensitivity of the system's dynamics.
In the context of a linear time-invariant input-output system these quantities are captured by the controllability and observability Gramians.
Thanks to their transformation properties, these Gramians may be simultaneously diagonalized by a change of coordinates, leading to the balanced truncation (BT) algorithm introduced by Moore \cite{Moore1981principal}.
Rowley \cite{Rowley2005model} discovered that the oblique projection and features extracted by BT can be approximated using a method of snapshots called balanced POD (BPOD) involving data obtained by impulse response simulations of the system and its adjoint. 
This enabled BT to be successfully applied to enormous systems including linearized fluid simulations where it would be computationally impractical to compute the Gramians \cite{Barbagallo2009closed, Ahuja2010feedback,
  Ilak2010model, Illingworth2011feedback}.
Analogous approaches can be applied to linear time-periodic input-output systems, where reduced-order models are obtained by balancing the time-periodic Gramians \cite{Sandberg2004balanced,Ma2010snapshot,MaThesis,Padovan2022bt}.

Moving away from linear systems, we can find several generalizations of balanced truncation for quadratic bilinear (QB) systems \cite{Benner2017balanced}, as well as to nonlinear systems \cite{Scherpen1993balancing, Verriest2004nonlinear, Lall2002subspace}.
However, there is room for improvement on these methods.
As we demonstrate in \cite{Otto2022optimizing}, QB balancing (QB BT) as well as the QB iterative rational Krylov algorithm (QB-IRKA) \cite{Benner2018H2} lose accuracy away from the stable equilibrium about which they are based in addition to being extremely difficult to compute for large-scale systems, e.g., with $10^5$ or more states.
The nonlinear balancing methods introduced in \cite{Scherpen1993balancing, Verriest2004nonlinear} are computationally challenging owing to the need to solve PDEs over the state space.
Progress has been made on this issue using Taylor series approximation \cite{Fujimoto2008computation, Kramer2022nlBT1, Kramer2023nlBT2}, but the computations remain daunting for systems with more than $10^3$ states.
The method introduced in \cite{Lall2002subspace} relies on empirical Gramians obtained via nonlinear impulse response simulations.
However, the choices for the impulses are ambiguous \cite{Ilak2010model, Ilak2009model}, and the number of impulses required scales linearly with the state dimension \cite{Kawano2021empirical}.
This scaling with the state dimension is unavoidable since the range of the empirical observability Gramian, measuring the system's sensitivity, lies in the span of the sampled initial conditions.
A generalization of the empirical Gramians framework is developed in \cite{Bouvrie2010balanced, Bouvrie2017kernel}, where state and output trajectory data are lifted into a RKHS, yielding a nonlinear dimension reduction map that can be computed using the kernel function.
This approach still relies on sampled trajectories for each state dimension of the FOM.

Sampling a function's gradient and applying POD to these samples reveals an ``active subspace'' \cite{Constantine2015active} in which most of the function's variation is captured.
Because the number of samples can be much smaller than the dimension of the function's domain, active subspaces have become a valuable tool for parameter studies and design optimization.
The generalization introduced by Zahm et al. \cite{Zahm2020gradient} identifies a projection that balances the sensitivity of an arbitrary function against the variance of an underlying probability distribution of points in the domain.
In the present work, we apply this approach to reduce the dimension of dynamical systems by considering the function that maps initial states to sequences of outputs.
We identify a projection that balances the sensitivity of this function with the variance of state space data collected from the system along typical trajectories.
The gradient samples are computed using an adjoint sensitivity method applied to random projections of output sequences from the system.
The resulting method, which we call covariance balancing reduction using adjoint snapshots (CoBRAS) is demonstrated on several challenging problems.

Our key technical observation is that the method introduced by Zahm et al. \cite{Zahm2020gradient} is identical to BT with the Gramians replaced by state and gradient covariance matrices sharing the same transformation properties.
Immediately this yields a new method of snapshots resembling BPOD that allows the projection to be computed for high-dimensional systems where assembling the covariance matrices is prohibitive.
This projection is optimal with respect to reconstruction of the state and gradient data in the sense described by Singler \cite{Singler2010optimality, Singler2015optimality}.
Moreover, the features extracted by the projection depend only on inner products among state and gradient vectors.
This allows us to replace the inner product with a kernel function in order to identify rich nonlinear features in a higher-dimensional RKHS, as in \cite{Scholkopf1998nonlinear}.
While related work by \cite{Romor2020kernel} introduces the idea of lifting into a high-dimensional space using a nonlinear feature map, the underling reliance on inner products enabling implicit computations via a kernel is not recognized.
Another related approach \cite{Bigoni2021nonlinear} optimizes nonlinear features based on gradient alignment.
In contrast, our approach does not require optimization and relies only on the singular value decomposition (SVD) of a kernel matrix whose size is independent of the state dimension.
Our approach also differs from the kernel-based generalization of the empirical Gramians framework described in \cite{Bouvrie2010balanced, Bouvrie2017kernel} since we rely on gradient samples obtained from the system's linearized adjoint and an appropriate lifting of these gradients into the RKHS in order to avoid the costly sampling of a trajectory for each state variable of the FOM.

We demonstrate the proposed method on a challenging three-dimensional nonlinear toy model, and on an incompressible axisymmetric jet flow simulation at Reynolds number $Re = 2000$. 
We will see that, not only does CoBRAS perform well as a standalone model reduction technique, but it can also be used to provide a good initial guess for the recently introduced trajectory-based optimization for oblique projections (TrOOP) framework \cite{Otto2022optimizing}.

\section{Balancing variance and sensitivity}

Consider a discrete-time dynamical system referred to as the full-order model (FOM)
\begin{equation}
\begin{aligned}
    x(t+1) &= f(x(t), u(t)), \qquad x(0) = x_0 \\
    y(t) &= g(x(t), u(t)),
\end{aligned}
\label{eqn:FOM}
\end{equation}
with state vector $x\in\R^n$, input $u\in\R^{q_0}$, and output $y \in \R^{m_0}$.
We seek to approximate the output of this system over a range of initial conditions and input signals using a reduced-order model (ROM)
\begin{equation}
\begin{aligned}
    z(t+1) &= \tilde{f}(z(t), u(t)), \qquad z(0) = h(x_0) \\
    \hat{y}(t) &= \tilde{g}(z(t), u(t)),
\end{aligned}
\label{eqn:ROM}
\end{equation}
whose state $z$ evolves in a lower-dimensional space $\R^r$ with $r < n$.
The ROM may be constructed by first selecting a map $h:\R^n \to \R^r$ and then employing one of several possible methods to find the approximate dynamics $\tilde{f}$ and observation map $\tilde{g}$ that govern the encoded state variables $z = h(x)$.
These methods include (nonlinear) Galerkin methods as well as regression-based approaches and recurrent neural networks.
Regardless of the approximation technique used for $\tilde{f}$ and $\tilde{g}$, the quality of the ROM depends heavily on the features selected by $h$.
Our goal will be to choose the map $h$ using simple data-driven methods that enable construction of accurate reduced-order models.
The map we identify will be associated with a projection operator that can be used to construct Petrov-Galerkin ROMs.

The features extracted by the encoding map $h$ must contain the information we need to forecast the future outputs of the system.
In this paper, we formalize this notion as follows. 
Evolving the system over a time-horizon $L$ with a sequence of inputs $u(0 \scolon L-1) = \big(u(0), \ldots, u(L-1)\big) \in \R^q$ produces outputs $y(0\scolon L) = \big( y(0), y(1), \ldots, y(L) \big) \in \R^m$, giving rise to a map
\begin{equation}
    F = ( F_0, \ldots, F_{L} ): \big(x(0),\ u(0\scolon L-1)\big) \mapsto y(0\scolon L),
    \label{eqn:input_output_map}
\end{equation}
with components $F_{\tau}: \big(x(0),\ u(0\scolon \tau-1)\big) \mapsto y(\tau)$.
In the present context, we seek an $h$ that allows for close approximation of $F\big(x(t),\ u(t\scolon t+L-1)\big)$ by a function $\tilde{F}\big(h(x(t)),\ u(t\scolon t+L-1)\big)$ along trajectories of the system.

Specifically, we consider approximating $F$ over a probability distribution of states $x$ and input sequences $\bar{u}$ constructed as follows.
We begin with a distribution over initial conditions $x_0$ and input sequences $u(0\scolon N+L-1)$ determined by the problem of interest.
From the resulting trajectories of \cref{eqn:FOM}, we draw $t$ uniformly at random from $\{0, \ldots, N \}$ to construct $x = x(t)$ and $\bar{u} = u(t\scolon t+L-1)$.

\subsection{Extracting linear features}
\label{subsec:linear_features}
We first consider the case where $h(x) = \Psi^T x$ is a linear map providing coordinates for the range of a linear projection $P = \Phi \Psi^T:\R^n \to \R^n$.
The projection uniquely decomposes any $x \in \R^n$ into $x = x_1 + x_2$, where $x_1 \in \Range(P)$ and $x_2 \in \Null(P) = \Null(h)$.
Intuitively speaking, a good projection for approximating $F$ will be one where the variance of $x_2$ is small and the sensitivity of $F$ to variation of $x_2$ is also small.
The problem of selecting a suitable projection based on these criteria for a general vector valued function $F$ is studied by Zahm et al. \cite{Zahm2020gradient}.
In this approach, one starts with a probability distribution over states and quantifies the variation of the state $x$ using the covariance matrix
\begin{equation}
    W_x = \E[x x^T].
    \label{eqn:state_covariance}
\end{equation}
Here, we apply this idea to a dynamical system \cref{eqn:FOM} by considering the probability distribution constructed above for states and input sequences along trajectories.
\begin{remark}
We center \cref{eqn:state_covariance} about the origin rather than the mean as one may wish to center about an arbitrary point in state space via a coordinate shift; for example, it is common to center about an equilibrium of the system \cref{eqn:FOM}.
\end{remark}

Following \cite{Zahm2020gradient}, the sensitivity of $F$ may be quantified using the gradient covariance matrix
\begin{equation}
    W_g = \E\left[ \grad_x F(x, \bar{u}) \grad_x F(x, \bar{u})^T  \right],
    \label{eqn:gradient_covariance}
\end{equation}
where $\grad_x F(x, \bar{u}) = \D_x F(x, \bar{u})^T$.
The optimal approximating function in the mean-square sense is given by the conditional expectation
\begin{equation}
    \hat{F}(P x, \bar{u}) = \E\left[ F(x, \bar{u}) \ \big\vert \ P x \right],
    \label{eqn:conditional_expectation_approximator}
\end{equation}
which averages $F(x_1 + x_2, \bar{u})$ over $x_2$, where $x = x_1 + x_2$ is the decomposition induced by $P$.
To understand the relationship between the covariance matrices and the approximation accuracy of \cref{eqn:conditional_expectation_approximator}, we consider the case where the conditional distribution of $x$ given $\bar{u}$ is Gaussian.
In this setting, \cref{thm:projection_error_for_Gaussian} provides an explicit bound.
\begin{theorem}
\label{thm:projection_error_for_Gaussian}
Let $(x, \bar{u})$ have distribution $\mu$ on $\R^n\times\R^q$ such that the conditional distribution of $x$ given $\bar{u}$ is almost surely Gaussian with positive-definite covariance $\Sigma_{x\vert\bar{u}} = \E\big[ \big( x - \E[x\vert\bar{u}] \big) \big( x - \E[x\vert\bar{u}] \big)^T \ \big\vert \ \bar{u} \big]$.
Let $\Sigma_x = \E\big[ \big( x - \E[x] \big) \big( x - \E[x] \big)^T \big]$ denote the marginal covariance.
Suppose that there is a constant $C \geq 0$ so that 
\begin{equation}
    C \Sigma_x - \Sigma_{x\vert\bar{u}} \succeq 0
\end{equation}
is positive semi-definite almost surely.
This holds with $C = 1$ when $(x, \bar{u})$ are jointly Gaussian.
Let $F \in L^2(\R^n\times\R^q, \mu; \R^m)$ have continuous partial derivatives and let $P:\R^n \to \R^n$ be a rank-$r$ linear projection.
Then the mean square approximation error of \cref{eqn:conditional_expectation_approximator} is bounded by
\begin{equation}
    \E\left[ \big\Vert F(x, \bar{u}) - \hat{F}(P x, \bar{u}) \big\Vert^2 \right] \leq C \Tr\Big[ W_g \big(I - P\big) W_x \big(I-P\big)^T \Big],
    \label{eqn:projection_error_for_Gaussian}
\end{equation}
where $W_x$ and $W_g$ are given by \cref{eqn:state_covariance} and \cref{eqn:gradient_covariance}.
\end{theorem}
\begin{proof}
The core of this result is Proposition~2.5 in Zahm et al. \cite{Zahm2020gradient}, which in turn is derived from the Gaussian Poincar\'{e} inequality in Chen \cite{Chen1982inequality}.
Since the rest of the proof is not especially instructive, we give it in \cref{app:projection_error_for_Gaussian}.
\end{proof}
Similar bounds likely hold for much larger classes of non-Gaussian distributions due to the results in \cite{Parente2020generalized, Zahm2022certified}.
We also note that the upper bound on the right-hand-side of \cref{eqn:projection_error_for_Gaussian} can be re-written as
\begin{equation}
    \Tr\Big[ W_g \big(I - P\big) W_x \big(I-P\big)^T \Big] 
    = \E\Big[ \big\Vert W_g^{1/2} (x - P x) \big\Vert^2 \Big].
\end{equation}
This can be interpreted as the mean square gradient-weighted difference between states and their projections.

In \cite{Zahm2020gradient}, a projection operator that minimizes the upper bound \cref{eqn:projection_error_for_Gaussian} is found by solving a generalized eigenvalue problem $W_g V = W_x^{-1} V \Lambda$ and constructing $P = V_r V_r^T W_x^{-1}$ where the columns of $V_r$ are the $r$ eigenvectors with largest eigenvalues (see Proposition~2.6 in \cite{Zahm2020gradient}).
The covariance matrices are approximated via Monte-Carlo sampling.
However, as the dimension $n$ becomes very large, it becomes computationally impractical to assemble the covariance matrices and solve the eigenvalue problem.
Moreover, as the output dimension $m$ becomes large, it becomes impractical to compute $\grad_x F$.

We address these issues using the observation that the projection found in \cite{Zahm2020gradient} is identical to the one found by applying balanced truncation (BT) \cite{Moore1981principal} to the state and gradient covariance matrices instead of the controllability and observability Gramians of a linear system.
The BT algorithm has been extensively studied and admits computationally efficient snapshot-based approximations in the case of large $n$ and $m$ using the balanced POD algorithm (BPOD) introduced by Rowley \cite{Rowley2005model}.
The key observation is \cref{thm:covariance_balancing}, which yields a factorized optimal projection using singular value decomposition (SVD) and factors of the covariance matrices.
This result is closely related to the optimality of BPOD for data reconstruction shown by Singler in \cite{Singler2010optimality}.
\begin{theorem}[Factorized covariance balancing]
\label{thm:covariance_balancing}
Let $W_x = X X^T$ and $W_g = Y Y^T$ and form the SVD
\begin{equation}
    Y^T X = U \Sigma V^T, \qquad \Sigma = \diag(\sigma_1, \ldots, \sigma_n),
    \label{eqn:balancing_SVD}
\end{equation}
with its rank-$r$ truncation denoted by $U_r \Sigma_r V_r^T$.
If $\sigma_r > 0$, then the minimum
\begin{equation}
    \min_{\substack{P \in \R^{n\times n} : \\ P^2 = P, \ \rank (P) = r }} \Tr\Big[ W_g \big(I - P\big) W_x \big(I-P\big)^T \Big] = \sigma_{r+1}^2 + \cdots + \sigma_n^2
\end{equation}
is achieved by
\begin{equation}
\boxed{
    P = \Phi \Psi^T, \quad \mbox{where} \quad
    \Phi = X V_r \Sigma_r^{-1/2} \quad \mbox{and} \quad \Psi = Y U_r \Sigma_r^{-1/2}.
    }
    \label{eqn:balancing_projection}
\end{equation}
\end{theorem}
\begin{proof}
Thanks to the permutation identity for the trace and the definition of the Frobenius norm, we have
\begin{equation}
    \Tr\Big[ W_g \big(I - P\big) W_x \big(I-P\big)^T \Big]
    = \big\Vert Y^T X - Y^T P X \big\Vert_F^2.
\end{equation}
The matrix $Y^T P X $ has rank at most $r$, and so the Eckart–Young–Mirsky theorem \cite{Eckart1936approximation, Mirsky1960symmetric} gives
\begin{equation}
    \big\Vert Y^T X - Y^T P X \big\Vert_F^2
    \geq \big\Vert Y^T X - U_r \Sigma_r V_r^T \big\Vert_F^2 = \sigma_{r+1}^2 + \cdots + \sigma_n^2.
\end{equation}
Direct substitution shows that equality is obtained with $P$ defined by \cref{eqn:balancing_projection}.
\end{proof}

To clarify the connection with balanced truncation, we observe that the state and gradient covariance matrices obey the same transformation laws as the controllability and observability Gramians under linear changes of coordinates $x = T \tilde{x}$.
The function to be approximated becomes $F(x) = F(T \tilde{x}) =: \tilde{F}(\tilde{x})$ and its derivative transforms according to $\D F(x) = \D \tilde{F}(\tilde{x}) T^{-1}$.
Thus, the state and gradient covariance matrices transform according to
\begin{equation}
    W_x = T \tilde{W}_x T^T \qquad \mbox{and} \qquad
    W_g = T^{-T} \tilde{W}_g T^{-1},
    \label{eqn:covariance_transformation_properties}
\end{equation}
where $W_g$ is computed using $F$ and $\tilde{W}_g$ is compute using $\tilde{F}$.
Thanks to this transformation law, when $W_x$ and $W_g$ are positive definite it is possible to find a ``balancing'' transformation $T = X V \Sigma^{-1/2}$, $T^{-1} = \Sigma^{-1/2} U^T Y^T$ so that $\tilde{W}_x = \tilde{W}_g = \Sigma$ are equal and diagonal \cite{Laub1987computation, Moore1981principal}.
It is easy to see that the projector described by \cref{eqn:balancing_projection} corresponds to the truncation operator
\begin{equation}
    T^{-1} P T = 
    \begin{bmatrix}
    I_r & 0 \\
    0 & 0
    \end{bmatrix}
\end{equation}
in the coordinate system where the covariance matrices are balanced.
While a balancing transformation requires the covariance matrices to be positive definite, the optimal projection described by \cref{thm:covariance_balancing} does not.
It may be computed using low-rank factors $X$ and $Y$.

Following a similar construction to BPOD \cite{Rowley2005model}, factors $X$ and $Y$ of the sample-based covariance matrices can be formed using the Monte-Carlo samples of the state and gradient.
Letting $\{x_1, \ldots, x_{s_x}\}$ be samples of the state, then the sample-based state covariance can be factored using
\begin{equation}
    X = \frac{1}{\sqrt{s_x}} \begin{bmatrix}
    x_1 & \cdots & x_{s_x}
    \end{bmatrix}.
    \label{eqn:state_sample_matrix}
\end{equation}
In order to reduce the computational burden associated with sampling the gradient, we use an independent, zero-mean, isotropic random vector $\xi\in\R^m$, i.e., $\E[\xi\xi^T] = I$ to define the univariate gradient
\begin{equation}
    g = \grad_x(\xi^T F)(x, \bar{u})
    \label{eqn:g-def}
\end{equation}
satisfying $\E[g g^T] = W_g$.
We can then factor the sample-based gradient covariance matrix using samples of the univariate gradient arranged into the columns of
\begin{equation}
    Y = \frac{1}{\sqrt{s_g}} \begin{bmatrix}
    g_1 & \cdots & g_{s_g}
    \end{bmatrix},
    \label{eqn:gradient_sample_matrix}
\end{equation}
where the columns $g_i$ are drawn using the joint distribution of $\xi$ and $x$
(which are independent).
Using these factors and \cref{thm:covariance_balancing}, we can compute the desired linear features using
\begin{equation}
\boxed{
    z = h(x) = \Psi^T x = \Sigma_r^{-1/2} U_r^T Y^T x.
    }
    \label{eqn:linear_feature_map}
\end{equation}
This computation does not use the $n\times n$ covariance matrices, and may be performed using inner products between sampled state and gradient vectors in $\R^n$.
In particular, assembling $Y^T X$ and computing its SVD has time complexity $\mathcal{O}(s_x s_g n + s_x s_g \min\{ s_x, s_g \} )$.
Once this is done, evaluating \cref{eqn:linear_feature_map} has time complexity $\mathcal{O}(s_g n + s_g r)$.

\subsection{The need for gradient sampling}
In this section we discuss why it is important to sample the gradient of $F$ rather than relying solely on samples of its values $F(x_i)$ at points $x_i\in\R^n$, $i=1, \ldots,s$.
Consider the simplest case when $F$ is a linear map with rank $r$.
The essential issue is that, if one relies solely on the sampled values $F(x_i)$, then generically, one cannot learn about the $r$-dimensional subspace on which $F$ is the most sensitive (i.e., the range of $F^T$) without sampling in all $n$ directions.
However, the subspace $\Null(F)^{\perp} = \Range(F^T) = \Range(W_g)$ is quickly spanned using $r$ samples of the gradient $\grad (\xi_i^T F)(x_i) = F^T \xi_i$ for almost every $\xi_1, \ldots, \xi_r$ (with respect to Lebesgue measure).
This well-known result is an immediate consequence of the following lemma.
\begin{lemma}
\label{lem:generic_sampling}
Let $M \in \R^{m \times n}$ be a matrix with $\rank(M) = r$.
If $r \leq s \leq n$, then almost every matrix $X \in \R^{n\times s}$ (with respect to Lebesgue measure) has linearly independent columns and satisfies $\rank(M X) = r$.
\end{lemma}
\begin{proof}
See \cref{app:classification_of_data_equivalent_subspaces}
\end{proof}
Thus, if columns of $X$ are the vectors $\xi_1,\ldots,\xi_r$, the $r$-dimensional range of $M=F^T$ is spanned by $F^T\xi_1,\ldots,F^T\xi_r$.
Related approximation results based on random sampling for matrices with decaying singular values, but possibly full rank, can be found in \cite{Halko2011finding}.

On the other hand, $\Range(F^T)$ cannot be uniquely determined from the pairs $\big( x_i, F(x_i) \big)$ until the number of samples $s$ is at least as large as $n-r$.
When $s < n-r$ there is always a map $\tilde{F}$ with $\tilde{F} x_i = F x_i$ for every $i=1, \ldots, s$, but with $\Range(\tilde{F}^T) \neq \Range(F^T)$.
To see this, when $s < n-r$ we cannot have $\Null(F) \subset \Range(X)$ since their dimensions are incompatible.
Considering the orthogonal complements, we cannot have $\Range(X)^{\perp} \subset \Range(F^T)$.
Given a (reduced) SVD $F = U\Sigma V^T$,
adding a nonzero vector $v \in \Range(X)^{\perp}$ with $v\notin \Range(F^T)$ to the first row of $V^T$ produces $\tilde{F} = U \Sigma (V^T + e_1 v^T)$ with the stated properties.
To quantify the degree of ``non-uniqueness'' for $\Range(F^T)$ given the samples, we recall that $r$-dimensional subspaces of $\R^n$ belong to the Grassmann manifold $\calG_{n,r}$ \cite{Wong1967differential, Absil2004riemannian, Bendokat2020grassmann}.
For generic samples forming the columns of an $n\times s$ matrix~$X$, the following \cref{thm:classification_of_data_equivalent_subspaces} characterizes the possible subspaces $\Range(\tilde{F}^T)$ as a submanifold of $\calG_{n,r}$ with dimension $r(n-s)$.

\begin{theorem}
\label{thm:classification_of_data_equivalent_subspaces}
Let $F:\R^{m\times n}$ be a matrix with $\rank(F) = r$ and let
$X\in \R^{n\times s}$ be a matrix with linearly independent columns satisfying $\rank(F X) = r$.
Then
\begin{equation}
    \calV_{F,X} = \left\{ \Range\big(\tilde{F}^T\big) \ : \ \tilde{F}\in\R^{m\times n},\ \rank(\tilde{F}) = r, \ \mbox{and} \ \tilde{F} X = F X \right\}.
\end{equation}
is an $r(n-s)$-dimensional submanifold of $\calG_{n,r}$ diffeomorphic to $\R^{(n-s)\times r}$.
The diameter of $\calV_{F,X}$ in $\calG_{n,r}$ is bounded by
\begin{equation}
    \frac{\pi}{2} \sqrt{\min\{r, n-s \}} \leq
    \sup_{V,W \in \calV_{F,X}} d(V,W) 
    \leq \sup_{V,W \in \calG_{n,r}} d(V,W)
    \leq \frac{\pi}{2}\sqrt{r},
\end{equation}
where $d$ denotes the geodesic distance on $\calG_{n,r}$ (see \cite{Wong1967differential, Bendokat2020grassmann}).
\end{theorem}
\begin{proof}
See \cref{app:classification_of_data_equivalent_subspaces}.
\end{proof}
Thus, if $s<n$ (so that the dimension of $\mathcal{V}_{F,X}$ is positive), then there are infinitely many possibilities for the range of $F^T$ that are consistent with the sampled data.
When $s \leq n-r$, these possibilities are not close together, in fact, they can differ by the maximum amount possible between subspaces.

In the general setting, the covariance matrices $W_x$ and $W_g$ may not have low rank.
However, in many cases they have quickly decaying eigenvalues giving them low ``effective rank''
\begin{equation}
    r(W) = \Tr(W) / \Vert W\Vert, 
\end{equation}
where $\Vert W\Vert$ is the operator norm (see Remark~5.53 in Vershynin \cite{Vershynin2012introduction}, as well as Sections 5.6, 7.6.1, and 9.2.3 in \cite{Vershynin2018high}).
For linear time-invariant systems this is known to hold for the observability and controllability Gramians when there are few inputs and observations \cite{Penzl1999cyclic, Baker2015fast}.
Results on the non-asymptotic theory of random matrices including \cite{Rudelson2007sampling, Oliveira2010sums, Vershynin2012introduction, Koltchinskii2017concentration, Holodnak2018probabilistic} show that covariance matrices with low effective rank can be estimated accurately in the operator norm with high probability using a number of samples that is independent of the ambient dimension $n$.
This allows for accurate computation of projectors \cref{eqn:balancing_projection} using state and gradient samples even when the state dimension far exceeds the number of samples we can reasonably obtain.

\subsection{Extracting nonlinear features using a kernel method}
\label{subsec:kernel_CoBRAS}
The fact that the features \cref{eqn:linear_feature_map} can be computed relying only on inner products in $\R^n$ suggests a natural reformulation as a kernel method yielding nonlinear features $z = h(x)$.
Here, we apply the same approach as above after lifting the problem into a high-dimensional (possibly infinite-dimensional) reproducing kernel Hilbert space (RKHS) and computing the required inner products implicitly via the reproducing kernel.
A similar approach is described in \cite{Romor2020kernel}, where feature maps into finite-dimensional spaces were computed explicitly.

We begin by summarizing some basic results and definitions that can be found in \cite{Berlinet2011reproducing, Hofmann2008kernel}.
Recall that a RKHS $\calH$ over a subset $\calX \subset\R^n$ is a Hilbert space of functions where point-wise evaluation of $f \in \calH$ at $x\in\calX$ is a bounded linear functional
\begin{equation}
    f(x) = \left\langle K_x,\ f \right\rangle_{\calH}.
\end{equation}
The function $K:\calX\times\calX \to \R$ defined by $K(x,y) := K_y(x) = \left\langle K_x,\ K_y \right\rangle_{\calH}$ is called the ``reproducing kernel'' of the RKHS.
The reproducing kernel is symmetric, namely $K(x,y) = K(y,x)$, and positive-definite in the sense that for any $x_1, \ldots, x_N \in \calX$ and constants $c_1, \ldots, c_N \in \R$ we have
\begin{equation}
    \sum_{i,j = 1}^N c_i K(x_i, x_j) c_j  \geq 0.
\end{equation}
Moreover, any function $K:\calX\times\calX \to \R$ with these properties uniquely defines a RKHS whose reproducing kernel is $K$.
In particular, the RKHS defined by $K$ is the completion of the span of $\{ K_x \}_{x\in\calX}$.
The ``feature map'' $\Phi_K : \calX \to \calH$ defined by $x\mapsto K_x$ provides lifted representatives of states in the RKHS.
When the kernel is continuous, then so is the feature map.
The so-called ``kernel trick'' in machine learning refers to working with such lifted representatives implicitly by means of the kernel $K$ without explicit calculations in $\calH$, which may be infinite-dimensional.

Observe that the empirical covariance matrix $W_x$ in \cref{eqn:state_covariance} may be defined by its action on a vector $v\in\R^n$ by
\[
W_x v = \frac{1}{s_x}\sum_{i=1}^{s_x}x_i \langle x_i, \ v \rangle.
\]
Analogously, the lifted empirical covariance operator $W_x':\calH\to\calH$ for a collection of states $\{x_i\}_{i=1}^{s_x} \subset \calX$ is defined by its action on a function $f\in\calH$ by
\begin{equation}
    W_x' f = \frac{1}{s_x} \sum_{i=1}^{s_x} \big(K_{x_i} - K_0\big)\big\langle K_{x_i} - K_0, \ f \big\rangle_{\calH},
    \label{eqn:lifted_state_covariance}
\end{equation}
where we are centering about the lift of the origin $K_0$.
Analogously to the finite-dimensional case, this operator can be factorized as $W_x' = X X^*$ with $X^*$ denoting the adjoint of the operator $X: \R^{s_x} \to \calH$ whose action on $w=(w_1, \ldots, w_{s_x})\in\R^{s_x}$ is defined by
\begin{equation}
    X w = \frac{1}{\sqrt{s_x}} \sum_{i=1}^{s_x} \big( K_{x_i} - K_0 \big) w_i.
    \label{eqn:lifted_X}
\end{equation}

We must define a gradient covariance operator for balancing by lifting the gradients $\grad (\xi^T F)(x)$ into the RKHS.
To do this, we ensure that there is a differentiable function $F'$ defined on the image $\calX' := \Phi_K(\calX) \subset \calH$ so that the diagram
\begin{equation}
    \begin{tikzcd}
        \calX \arrow[r, "F"] \arrow[d, "\Phi_K"'] & \R^m \\ 
        \calX' \arrow[ur, "F'"'] &
    \end{tikzcd}
    \label{cd:function_lifting}
\end{equation}
commutes.
We lift the gradient $\grad (\xi^T F)(x)$ into $\calH$ by computing $\grad(\xi^T F')(K_x)$.
Since this will involve differentiating the feature map, we briefly summarize a remarkable result by Zhou \cite{Zhou2008derivative} showing that the derivatives of a smooth kernel also have reproducing properties.
Let $\alpha = (\alpha_1, \ldots, \alpha_n) \in \mathbb{N}_0^n$ be a multi-index and denote $\partial^{\alpha} = \frac{\partial^{\vert\alpha\vert}}{\partial x_1^{\alpha_1} \cdots \partial x_n^{\alpha_n}}$, where $\vert \alpha \vert = \alpha_1 + \cdots + \alpha_n$.
If $\calX$ is a bounded open set and $K\in C^{2s}(\bar{\calX}\times \bar{\calX})$, $s\geq 1$, is the reproducing kernel for $\calH$, then Theorem~1 in \cite{Zhou2008derivative} shows that for $\vert \alpha \vert \leq s$ the function $(\partial^{\alpha} K)_x$ defined by $(\partial^{\alpha} K)_x(y) := \partial^{(\alpha, 0)} K(x,y) = \partial^{\alpha} K_y(x)$ is an element of $\calH$ and
\begin{equation}
    \partial^{\alpha} f(x) = \left\langle (\partial^{\alpha} K)_x,\ f \right\rangle_{\calH}.
    \label{eqn:derivative_reproducing_property}
\end{equation}
In the following lemma, we use this result to show that the feature map is continuously Fr\'{e}chet differentiable.
\begin{lemma}
    \label{lem:feature_map_differentiability}
    Let $\calX \subset \R^n$ be a bounded open set and let $\calH$ be an RKHS with reproducing kernel $K\in C^2(\bar{\calX} \times \bar{\calX})$.
    Then the feature map $\Phi_K: x \mapsto K_x$ has a Fr\'{e}chet derivative $\D \Phi_K(x):\R^n \to \calH$ at each $x\in\calX$ and satisfies
    \begin{equation}
        \D \Phi_K(x) v = \sum_{j=1}^n (\partial^{e_j} K)_x v_j, \qquad
        \D \Phi_K(x)^* f = \grad f(x),
        \label{eqn:derivative_of_feature_map}
    \end{equation}
    for every $v=(v_1, \ldots, v_n)\in \R^n$ and $f \in \calH$.
    Moreover, the map $x \mapsto \D \Phi_K(x)$ with $x\in\calX$ is continuous with respect to the operator norm, i.e., $\Phi_K$ is a $C^1$ function on $\calX$ (see Definition~1.1.2 in Kesavan \cite{Kesavan2022nonlinear}).
\end{lemma}
\begin{proof}
See \cref{app:gradient_lifting}
\end{proof}
Given this setup, our main result is the following:
\begin{theorem}
    \label{thm:gradient_lifting}
    Let $\calX \subset \R^n$ be a bounded open set and let $\calH$ be an RKHS with reproducing kernel $K\in C^2(\bar{\calX} \times \bar{\calX})$ so that the feature map $\Phi_K$ is injective on the closure $\bar\calX$
    and the derivative Gram matrix,
    \begin{equation}
        G(x) := 
        \D \Phi_K(x)^* \D \Phi_K(x) =
        \begin{bmatrix}
            \partial^{(e_1, e_1)}K(x,x) & \cdots & \partial^{(e_1, e_n)}K(x,x) \\
            \vdots & \ddots & \vdots \\
            \partial^{(e_n, e_1)}K(x,x) & \cdots & \partial^{(e_n, e_n)}K(x,x)
        \end{bmatrix},
        \label{eqn:kernel_derivative_Gram_matrix}
    \end{equation}
    is positive-definite for every $x\in\calX$.
    Then the feature map $\Phi_K$ is a $C^1$ embedding of $\calX$ into $\calH$ (see Definition~3.3.1 in Margalef-Roig and Dominguez \cite{Margalef1992differential}).
    Hence, there is a $C^1$ function $F'=F \circ \Phi_K^{-1}:\Phi_K(\calX)\to\R^m$ so that \cref{cd:function_lifting} commutes and
    for every $x\in\calX$ we have
    \begin{equation}
        \grad ( \xi^T F' )(K_x) = \D \Phi_K(x) G(x)^{-1} \grad (\xi^T F)(x) \in \calH.
        \label{eqn:lifted_gradient}
    \end{equation}
\end{theorem}
\begin{proof}
See \cref{app:gradient_lifting}
\end{proof}
Note that the theorem requires that the feature map $\Phi_K:x\mapsto K_x$ be injective.  An easy test for injectivity is that, for every distinct $x_1,x_2\in \bar\calX$,    \begin{equation}
    \big\Vert K_{x_1} - K_{x_2} \big\Vert_{\calH}^2 =
    K(x_1,x_1) - 2 K(x_1, x_2) + K(x_2, x_2) > 0.
    \label{eqn:kernel_is_injective}
\end{equation}

To form the empirical gradient covariance operator, we first produce gradient samples $\{ g_i \}_{i=1}^{s_g}$ from randomly chosen
directions $\xi_i$ and states $\tilde x_i$, as in equation~\cref{eqn:g-def} (here,
we use $\tilde x_i$ for the sampled states, because these may be distinct
from the states $x_i$ used to form the state covariance operator
in~\cref{eqn:lifted_state_covariance}). We then lift the gradient samples using \cref{thm:gradient_lifting}, to define the empirical gradient covariance operator
\begin{equation}
    W_g' = \frac{1}{s_g} \sum_{i=1}^{s_g} \D \Phi_K(\tilde x_i) G(\tilde x_i)^{-1} g_i g_i^T G(\tilde x_i)^{-1} \D \Phi_K(\tilde x_i)^*.
\end{equation}
Analogously to the finite-dimensional case, this operator may be factorized as $W_g' = Y Y^*$ with $Y:\R^{s_g} \to \calH$ defined by its action on $w=(w_1, \ldots, w_{s_g})\in\R^{s_g}$ according to
\begin{equation}
    Y w = \frac{1}{\sqrt{s_g}}\sum_{i=1}^{s_g} w_i \D \Phi_K(\tilde x_i) G(\tilde x_i)^{-1} g_i.
    \label{eqn:lifted_Y}
\end{equation}
The adjoint of this operator acting on a lifted state $K_x$ for $x\in\calX$ is the vector in $\R^{s_g}$ whose $i$th component is given by
\begin{equation}
\boxed{
    \big[Y^* K_x\big]_i = \frac{1}{\sqrt{s_g}} g_i^T G(\tilde x_i)^{-1} \grad K_x(\tilde x_i).
    }
    \label{eqn:gradient_state_RKHS_product}
\end{equation}
Here, the $n$-dimensional vector
\begin{equation}
    \grad K_y(x)
    = \left( \partial^{(e_1,0)}K(x,y),\ \ldots,\ \partial^{(e_n,0)}K(x,y) \right)
\end{equation}
is found by differentiating the kernel with respect to the coordinates of the first entry.
Thus, we make the crucial observation that \cref{eqn:gradient_state_RKHS_product} can be computed using the kernel without doing explicit calculations in the possibly infinite-dimensional RKHS.

To find a truncated balancing transformation for $W_x'$ and $W_g'$, we compute an SVD of the $s_g\times s_x$ matrix
$
    Y^* X
    = U \Sigma V^T,
$
whose elements are given by
\begin{equation}
\boxed{
    [Y^*X]_{i,j} = \frac{1}{\sqrt{s_g s_x}} g_i^T G(\tilde x_i)^{-1} \big( \grad K_{x_j}(\tilde x_i) - \grad K_{0}(\tilde x_i) \big),
    }
\end{equation}
thanks to \cref{eqn:gradient_state_RKHS_product}.
As in the finite-dimensional case, an oblique projection $P:\calH \to \calH$ onto an $r$-dimensional subspace may be defined by \cref{eqn:balancing_projection}.
The information extracted by $P$ acting on a lifted state $K_x$ is encoded in the $r$-dimensional feature vector
\begin{equation}
\boxed{
    z = h(x)
    = \Psi^*(K_x - K_0)
    = \Sigma_r^{-1/2} U_r^T \big(Y^*K_x - Y^*K_0 \big).
    }
    \label{eqn:nonlinear_features}
\end{equation}
These features can be computed explicitly using the kernel and \cref{eqn:gradient_state_RKHS_product}.
The feature vectors associated with the original data are given by the columns of
\begin{equation}
    \begin{bmatrix}
    h(x_1) & \cdots & h(x_{s_x})
    \end{bmatrix} = \sqrt{s_x} \Sigma_r^{1/2} V_r^T.
\end{equation}
Although we do not pursue nonlinear Galerkin modeling in this paper, one could use the derivative of the feature vector \cref{eqn:nonlinear_features} for this purpose.
The action of the derivative on a vector $v\in\R^n$ is given by
\begin{equation}
    \D h(x) v = \Sigma_r^{-1/2} U_r^T \D \big( Y^* \Psi_K(x) \big) v,
    \label{eqn:kernel_embedding_tangent_map}
\end{equation}
where the elements of $\D \big( Y^* \Psi_K(x) \big) v =  \frac{1}{\sqrt{s_g}} \big[ g_i^T G(\tilde x_i)^{-1} H(\tilde x_i, x) v \big]_{i=1}^{s_g}$ are computed using the matrix $H(x,y) = \D \Phi_K(x)^*\D \Phi_K(y) = \big[ \partial^{(e_i, e_j)} K(x,y) \big]_{i,j=1}^n$.


Some examples of common kernels and the corresponding functions $\grad K_y(x)$ and $G(x)^{-1}$ are given in \cref{tab:common_kernels}.
Fortunately, we see that $\grad K_y(x)$ can be computed with time complexity $\mathcal{O}(n)$.
Moreover, in each case $G(x)^{-1}$ has structure that allows us to act with it on a vector with time complexity $\mathcal{O}(n)$. 
These properties allow the method to be implemented on problems with very large $n$, where matrix inversion and even dense matrix-vector products would be computationally prohibitive. 
Finally, we note that when $K(x,y) = x^T y$, the kernel method becomes identical to the technique described in \cref{subsec:linear_features}.

\begin{table}[]
    \centering
    \caption{Some common smooth kernels and the functions required to implement the kernel method.}
    \begin{tabular}{|c|c|c|}
        \hline
        $K(x,y)$ & $\grad K_y(x)$ & $G(x)^{-1}$ \\
        \hline
        $\alpha + x^T y, \quad \alpha \geq 0$ & $y$ & $I$ \\
        $(\alpha + x^T y)^p, \quad \left\{\substack{p > 1 \\ \alpha > 0}\right.$ & $p (\alpha+x^T y)^{p-1} y$ & $\frac{1}{p (\alpha + \Vert x\Vert^2)^{p-1}} \left[ I - \left( \frac{p-1}{\alpha + p \Vert x\Vert^2} \right) x x^T \right]$ \\
        $\exp\left( -\frac{\Vert x - y\Vert^2}{2 \sigma^2} \right)$ & $-\frac{1}{\sigma^2} K(x,y)(x-y)$ & $\sigma^2 I$ \\
        \hline
    \end{tabular}
    \label{tab:common_kernels}
\end{table}

\section{Randomized adjoint sampling methods}
\label{sec:adjoint_sampling}

In order to obtain gradient samples described in \cref{subsec:linear_features}, we must be able to compute gradients of the output sequence map~$F$.  Perhaps the most natural way to compute such gradients is by finite difference approximation; however, this method would scale poorly with the state dimension.  Instead, in this section we describe a more computationally efficient way of obtaining these gradient samples, using linearized adjoint equations derived from \cref{eqn:FOM}.
Our method entails computing gradients of output sequences with respect to the initial conditions $x_0$ of ``mini-trajectories'' with length $L+1$.
In computing these gradients, the adjoint method also produces gradients with respect to intermediate states along these trajectories, which, in general, must be discarded.
However, when we sample from the probability distribution constructed from longer trajectories of the system in the paragraph preceding \cref{subsec:linear_features} the intermediate gradient information can be retained, yielding computational benefits.
We discuss a gradient sampling method for an empirical version of this distribution in \cref{subsec_long_trajectories}.
In \cref{subsec_stationary_distribution} we discuss a simplified method that can be applied when the distribution over pairs $(x, \bar{u})$ is ``stationary'' in a sense that will be described later.

The random vectors $\xi\in \R^{m_0(L+1)}$ that we use to sample the gradient are constructed as follows.
Let $\eta \in \R^{m_0}$ be a zero mean random vector with $\E[\eta \eta^T] = (L+1) I$ and let $\tau \in \{0, \ldots, L \}$ be chosen uniformly at random.
Then the random vector $\xi \in \R^{m_0(L+1)}$ formed by placing $\eta$ into the $\tau$th slot of 
\begin{equation}
    \xi = \begin{pmatrix}
    0 & \cdots & 0 & \eta & 0 & \cdots & 0
    \end{pmatrix}
    = e_{\tau}\otimes \eta,
    \label{eqn:random_vector}
\end{equation}
also has zero mean and covariance $\E[\xi \xi^T] = I$.

Along a trajectory of \cref{eqn:FOM}, the gradients
\begin{equation}
g_{\eta}(t,k) = \grad_x (\eta^T F_{k})\big(x(t), u(t\scolon t+k-1)\big)
\label{eqn:gradients_defn}
\end{equation}
satisfy the adjoint equation
\begin{equation}
\begin{aligned}
    g_{\eta}(t_f-k, k) &= \D_x f\big( x(t_f-k), u(t_f -k) \big)^T g_{\eta}(t_f-(k-1), k-1) \\
    g_{\eta}(t_f,0) &= \D_x g(x(t_f))^T \eta.
\end{aligned}
    \label{eqn:adjoint_equation}
\end{equation}
With our choice of $\xi$ given by \cref{eqn:random_vector}, we generate
\begin{equation}
    \grad_x (\xi^T F)\big(x(0), u(0\scolon L-1) \big) 
    = \grad_x (\eta^T F_{\tau})\big( x(0), u(0\scolon \tau-1) \big)
    = g_{\eta}(0,\tau)
    \label{eqn:gradient_sample}
\end{equation}
by computing $g_{\eta}(\tau-k, k)$ for $k = 0, 1,\ldots, \tau$ recursively using \cref{eqn:adjoint_equation}.
In what follows, we discuss cases where these intermediate gradient samples can be incorporated in the empirical gradient covariance.

\subsection{Sampling from long trajectories}
\label{subsec_long_trajectories}
In many cases of practical interest we generate one or more long trajectories of snapshots $x(0\scolon N+L)$ via simulation with different initial conditions and input sequences $u(0\scolon N+L-1)$.
For simplicity, we consider a single long trajectory where
the first $N+1$ snapshots serve as initial conditions for the time-$L$ input-output map \cref{eqn:input_output_map}.
Hence, the state covariance is based on a uniform distribution over the first $N+1$ snapshots.
To draw samples of the gradient with respect to this distribution, we provide \cref{alg:gradient_sampling_long_trajs}.

\begin{algorithm}
\caption{Sample gradients from long trajectories}
\label{alg:gradient_sampling_long_trajs}
\begin{algorithmic}[1]
\STATE{\textbf{input}: the time-horizon $L$, the number of samples $s_g$, a long trajectory $\{ x(0), x(1), \ldots, x(N+L) \}$, and a distribution for $\eta\in\R^{m_0}$ with zero mean and $\E[\eta\eta^T] = (L+1)I$}
\FOR{$i=1, 2, \ldots, s_g$}
    \STATE{Draw $t'$ uniformly from $\{ 0, \ldots, N \}$, draw $\tau'$ uniformly from $\{ 0,\ldots, L \}$, and draw $\eta$ from its distribution.}
    \STATE{Letting $t_f = t' + \tau'$, solve the adjoint equation \cref{eqn:adjoint_equation} to generate $g_{\eta}(t_f - k, k)$ for each $k=0, \ldots, \min\{ L, t_f \}$.}
    \STATE{With $\tau_{\text{min}} = \max\{0, t_f-N\}$ and $\tau_{\text{max}} = \min\{ L, t_f \}$, arrange the samples with $\tau_{\text{min}} \leq k \leq \tau_{\text{max}}$ into the matrix \begin{equation*} 
    Y_i = 
    \frac{1}{\sqrt{1+\tau_{\text{max}}-\tau_{\text{min}}}} 
    \begin{bmatrix}
        g_{\eta}(t_f-\tau_{\text{min}}, \tau_{\text{min}}) 
        & \cdots 
        & g_{\eta}(t_f-\tau_{\text{max}}, \tau_{\text{max}})
    \end{bmatrix}.
    \end{equation*}}
\ENDFOR
\RETURN{the gradient sample matrix $Y = \frac{1}{\sqrt{s_g}}\begin{bmatrix}Y_1 & \cdots & Y_{s_g} \end{bmatrix}$.}
\end{algorithmic}
\end{algorithm}

During the $i$th stage of this procedure we obtain samples by solving the adjoint equation \cref{eqn:adjoint_equation} over a time horizon of length $\min\{L, t_f \}$.
In most applications, the computational cost to act with $\D_x f(x, u)^T$ on a vector is comparable to evaluating $f(x,u)$, meaning that the cost of each stage is comparable to simulating the FOM \cref{eqn:FOM} over the same time horizon $\min\{L, t_f \}$.

In what follows, we show that \cref{alg:gradient_sampling_long_trajs} corresponds to a Monte-Carlo approximation for the gradient covariance.
First, we observe that the gradient covariance matrix can be written in terms of the sequences generated by solving \cref{eqn:adjoint_equation} from final times $t_f = t' + \tau'$ with $(t',\tau')$ drawn uniformly from $\{0, \ldots, N\}\times\{0, \ldots, L\}$.
In particular, given a final time $t_f$ between $0$ and $N+L$, the number of initial times $t$ and prediction horizons $\tau$ that sum to $t_f$ is given by
\begin{equation}
\begin{aligned}
    \nu(t_f) =  \sum_{t=0}^N \sum_{\tau = 0}^L \delta_{t+\tau, t_f} 
    &= 1 + \min\{ t_f,\ N,\ L,\ N+L-t_f \} \\
    &= 1 + \min\{ L, t_f \} - \max\{0, t_f-N\}.
\end{aligned}
\label{eqn:final_time_counter}
\end{equation}
Using this counting function we can express the gradient covariance as
\begin{multline}
     W_g 
    = \frac{1}{(N+1)(L+1)} \sum_{t=0}^N \sum_{\tau = 0}^L \E_{\eta}\left[ g_{\eta}(t, \tau) g_{\eta}(t, \tau)^T \right] \\
    = \frac{1}{(N+1)(L+1)} \sum_{t'=0}^N \sum_{\tau' = 0}^L \E_{\eta}\left[ \frac{1}{\nu(t'+\tau')} \sum_{t=0}^N \sum_{\tau = 0}^L \delta_{t+\tau, t'+\tau'} g_{\eta}(t, \tau) g_{\eta}(t, \tau)^T \right] \\
    = \frac{1}{(N+1)(L+1)} \sum_{t'=0}^N \sum_{t_f = t'}^{t'+L} \E_{\eta}\left[ \frac{1}{\nu(t_f)} \sum_{k = \max\{0, t_f-N\}}^{\min\{ L, t_f \}} g_{\eta}(t_f-k, k) g_{\eta}(t_f-k, k)^T \right].
    \label{eqn:gradient_covariance_from_long_traj}
\end{multline}
The outer two summations in the last expression compute an average over final times $t_f$ and the expectation is computed with respect to the random vector $\eta$.
The inner summation in the last expression computes an average over the horizon lengths $k$ corresponding to initial times $t = t_f-k$ falling between $0$ and $N$.
The empirical covariance with factor $Y$ produced by \cref{alg:gradient_sampling_long_trajs} corresponds to using a Monte-Carlo method to approximate the outer two summations and the expectation over $\eta$ in the last expression of \cref{eqn:gradient_covariance_from_long_traj}.


\subsection{Sampling from a stationary distribution}
\label{subsec_stationary_distribution}
We consider a statistically stationary case where the distributions of states
and inputs are independent of time.  
More precisely, the distributions of $\big( x(t), u(t\scolon t+L-1) \big)$ and $\big( x_0, u(0\scolon L-1) \big)$ are assumed to be identical for each $t=1, \ldots, N$.
Hence, the distribution for $(x, \bar{u})$ described in the paragraph before \cref{subsec:linear_features} is the same as the distribution for $\big( x_0, u(0\scolon L-1) \big)$.
For example, this situation occurs when the initial conditions $x_0$ are sampled from the invariant distribution on an attractor with zero input, or input provided by state feedback plus independent noise.
Another example is when $x_0$ and $u$ are drawn by choosing an initial time uniformly along a periodic orbit of \cref{eqn:FOM}.
We observe that the stationarity assumption implies that $g_{\eta}(L-k,k)$ and $g_{\eta}(0,k)$ are identically distributed for $k=0, \ldots, L$, yielding
\begin{equation}
    W_g 
    = \E \left[ \frac{1}{L+1} \sum_{\tau=0}^L g_{\eta}(0,\tau) g_{\eta}(0,\tau)^T  \right]
    = \E \left[ \frac{1}{L+1} \sum_{k=0}^L g_{\eta}(L-k,k) g_{\eta}(L-k,k)^T   \right].
\end{equation}
Here the expectation is taken over the initial condition $x_0$, the input sequence $u(0:L-1)$ and the random vector $\eta$.

To construct a Monte-Carlo approximation of the gradient covariance, we first sample ``mini-trajectories'' $x_i(0\scolon L)$, $i=1, \ldots, s_g$ of the system.
For each, we choose $\eta_i$ independently at random and solve \cref{eqn:adjoint_equation}, yielding the columns of a matrix
\begin{equation}
    Y_i = \frac{1}{\sqrt{L+1}} 
    \begin{bmatrix}
        g_{\eta_i}(L,0) & g_{\eta_i}(L-1,1) & \cdots & g_{\eta_i}(0,L)
    \end{bmatrix}.
\end{equation}
The empirical gradient covariance for the collection of mini-trajectories can then be factored using
\begin{equation}
    Y = \frac{1}{\sqrt{s_g}} \begin{bmatrix}
        Y_1 & \cdots & Y_{s_g}
    \end{bmatrix},
    \label{eqn:gradient_covariance_factor_multiple_minitrajectories}
\end{equation}
where no intermediate gradient samples have been wasted.

\section{Results}

\subsection{A challenging model problem}
We consider the same problem proposed in \cite{Otto2022optimizing}, where we seek two-dimensional Petrov-Galerkin ROMs for the system
\begin{equation}
\begin{split}
    \dot{x}_1 &= -x_1 + 20 x_1 x_3 + u \\
    \dot{x}_2 &= -2 x_2 + 20 x_2 x_3 + u \\
    \dot{x}_3 &= -5 x_3 + u \\
    y &= x_1 + x_2 + x_3.
\end{split}
\label{eqn:toy_model}
\end{equation}
As discussed in \cite{Otto2022optimizing}, model reduction for this system is challenging due to its strong nonlinear interactions involving the state $x_3$, which has small variance compared to $x_1$ and $x_2$.
In \cite{Otto2022optimizing} the projections found by various
methods mentioned in \cref{sec:introduction} (namely, POD, BT, QB BT, and
QB-IRKA) were compared to those found by a newly proposed method: trajectory-based optimization
for oblique projections (TrOOP).
TrOOP is an iterative method that uses gradient descent to find a
Petrov-Galerkin projection that minimizes error along a collection of training trajectories.
Here, we add CoBRAS to this list using precisely the same setup.
The training data consisted of the two trajectories shown in \cref{fig:toy_model_training_trajectories} sampled every $\Delta t = 0.5$.
These are nonlinear impulse-responses generated by simulating \cref{eqn:toy_model} with zero input and initial conditions $x(0) = (u_0,u_0,u_0)$ with magnitudes $u_0 = 0.5$ and $u_0 = 1.0$.
The covariance matrices for CoBRAS were defined using $L=5$ of these intervals as the horizon length and initial conditions uniformly distributed over the $22$ sample points.

The prediction accuracy of the resulting ROMs on $100$ impulse responses with magnitudes $u_0$ drawn uniformly from the interval $[0,1]$ is shown in \cref{fig:toy_model_testing_performance}.
Performance on a trajectory with sinusoidal input $u(t) = \sin(t)$ is also shown in \cref{fig:toy_model_sinusoidal}.
We observe that CoBRAS achieves prediction accuracy comparable to the optimized projection found by TrOOP.
The performance is also not very sensitive to the choice of horizon length, with comparable prediction accuracy observed when $L \geq 4$.

\begin{figure}
    \centering
    \subfloat[training trajectories\label{fig:toy_model_training_trajectories}]{
    \begin{tikzonimage}[trim=20 10 40 20, clip=true, width=0.455\textwidth]{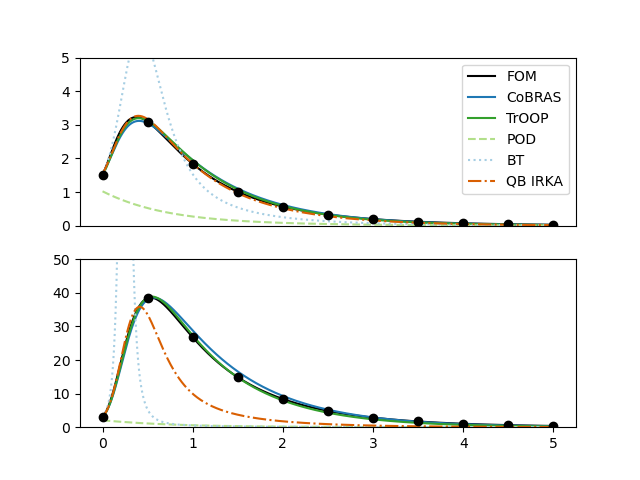}
    \node[rotate=90] at (0.0, 0.73) {\footnotesize $y$ for $u_0 = 0.5$};
    \node[rotate=90] at (0.0, 0.285) {\footnotesize $y$ for $u_0 = 1$};
    \node[rotate=0] at (0.55, 0.01) {\footnotesize Time $t$};
    %
    \end{tikzonimage}
    }
    \subfloat[error on testing trajectories\label{fig:toy_model_testing_performance}]{
    \begin{tikzonimage}[trim=10 10 40 20, clip=true, width=0.465\textwidth]{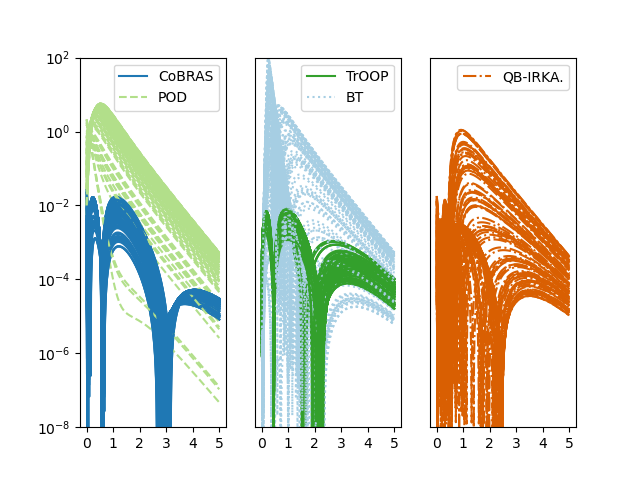}
    \node[rotate=90] at (0.0, 0.5) {\footnotesize $(\hat{y} - y)^2 / \avg{(y^2)}$};
    \node[rotate=0] at (0.245, 0.01) {\footnotesize Time $t$};
    \node[rotate=0] at (0.55,  0.01) {\footnotesize Time $t$};
    \node[rotate=0] at (0.855, 0.01) {\footnotesize Time $t$};
    \end{tikzonimage}
    }
    \caption{In (a) we show the output trajectories of the FOM \cref{eqn:toy_model} and various ROMs along the training trajectories. The samples are indicated by black dots. In (b) we show the normalized square prediction errors for each ROM along 100 trajectories with $u_0$ chosen uniformly at random from the interval $[0,1]$.}
    \label{fig:toy_model_impulse}
\end{figure}

\begin{figure}
    \centering
    \begin{tikzonimage}[trim=40 0 40 10, clip=true, width=0.7\textwidth]{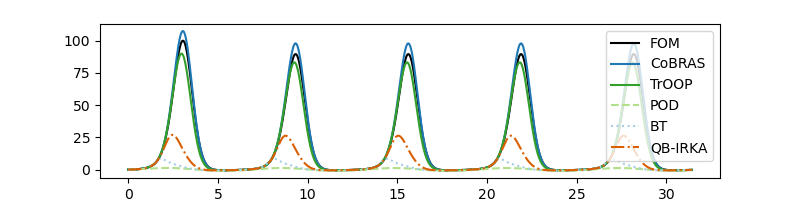}
    \node[rotate=90] at (0.0, 0.5) {\footnotesize Output $y$};
    \node[rotate=0] at (0.5, -0.05) {\footnotesize Time $t$};
    \end{tikzonimage}
    \caption{We show the output trajectory of the FOM \cref{eqn:toy_model} and the predicted outputs of each ROM with zero initial condition and input $u(t) = \sin(t)$.}
    \label{fig:toy_model_sinusoidal}
\end{figure}

\subsection{Nonlinear axisymmetric jet flow}
\label{sec:jet_flow}
We consider the same nonlinear jet flow problem discussed in \cite{Otto2022optimizing} at a higher Reynolds number $Re=2000$.
This system is governed by a discretization of the axisymmetric incompressible Navier-Stokes equations with $10^5$ states.
The flow is sensitive to disturbances introduced upstream in the shear layer about a stable equilibrium.
We observe the full state of the system given by the equilibrium-subtracted velocity field scaled by a radially-dependent weight.
The weight is chosen so that the inner product on the state space is Euclidean
and the squared norm of a state equals the kinetic energy of the corresponding equilibrium-subtracted flowfield.
Input is provided to the system via a source term in the radial momentum equation concentrated in the shear layer near a radius $0.5$ and down-stream distance from the nozzle $1.0$, as indicated in \cref{fig:Galerkin_snap_1}.
The resulting FOM can be written as
\begin{equation}
    \dot{x} = f_0(x) + b u, \qquad x(0) = x_0,
    \label{eqn:jet_flow_FOM}
\end{equation}
where $f_0$ is a quadratic function derived from our spatial discretization of the Navier-Stokes equations and $b$ is a vector corresponding to the source term described above.
For more details, see \cite{Otto2022optimizing}.

The training data consists of $12$ nonlinear impulse-response trajectories in which we set $u=0$ and $x_0 = b u_0$, with magnitudes $u_0 = \pm 0.005, \pm 0.02, \pm 0.04, \pm 0.06, \pm 0.08$, and $\pm 0.10$.
The state along each trajectory was sampled $100$ times at intervals $\Delta t = 0.5$.
The testing data consisted of $25$ such trajectories with impulse magnitudes drawn uniformly at random from the interval $[-0.1, 0.1]$.
The energy along these trajectories, plotted in \cref{fig:jet_training_trajectories}, indicates that the system undergoes rapid and nonlinear transient growth before the disturbances leave the computational domain through the outflow boundary.
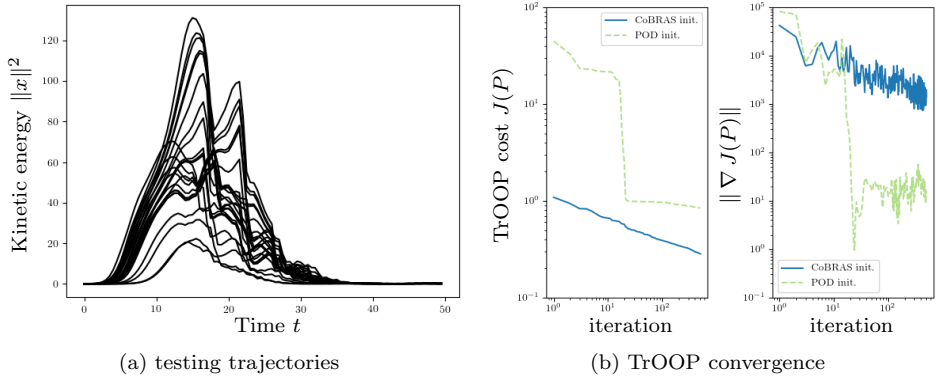
\begin{figure}
    \centering
    \subfloat[testing trajectories \label{fig:jet_training_trajectories}]{
    \begin{tikzonimage}[trim=20 10 40 20, clip=true, width=0.45\textwidth]{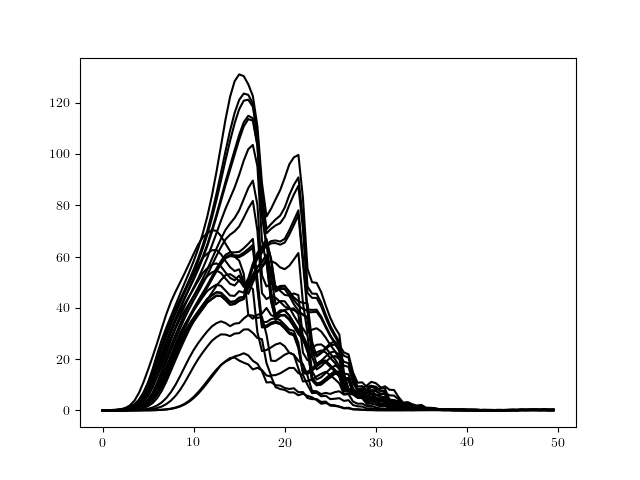}
    \node[rotate=90] at (-0.01, 0.5) {\footnotesize Kinetic energy $\Vert x \Vert^2$};
    \node[rotate=0] at (0.55, 0.01) {\footnotesize Time $t$};
    \end{tikzonimage}
    }
    \subfloat[TrOOP convergence \label{fig:TrOOP_convergence}]{
    \begin{tikzonimage}[trim=0 5 40 10, clip=true, width=0.45\textwidth]{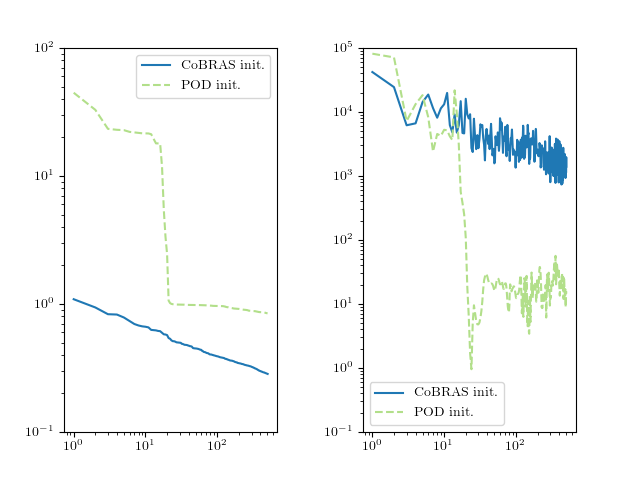}
    \node[rotate=90] at (0.01, 0.5) {\footnotesize TrOOP cost $J(P)$};
    \node[rotate=0] at (0.30, 0.01) {\footnotesize iteration};
    \node[rotate=90] at (0.525, 0.5) {\footnotesize $\Vert \grad J(P) \Vert$};
    \node[rotate=0] at (0.80, 0.01) {\footnotesize iteration};
    \end{tikzonimage}
    }
    \caption{In (a) we show the energy of the jet flow along the testing impulse-response trajectories. In (b) we compare the convergence of TrOOP for the jet flow using POD and CoBRAS for initialization and the cost function \cref{eqn:TrOOP_cost_fun}.}
    \label{fig:jet_energy_and_TrOOP_convergence}
\end{figure}

\subsubsection{Petrov-Galerkin models}
\label{subsubsec:jet_Galerkin}
We construct Petrov-Galerkin ROMs for the jet flow \cref{eqn:jet_flow_FOM} using $r=40$-dimensional projection operators $P = \Phi \Psi^T$ determined by POD, BPOD about the equilibrium, CoBRAS, and TrOOP.
Specifically, the state forecasts $\hat{x}(t) = \Phi z(t)$ were obtained by simulating the $40$-dimensional system
\begin{equation}
    \dot{z} = \tilde{f}(z,u) := \Psi^T f_0(\Phi z) + \Psi^T b u, \qquad z(0) = \Psi^T x_0,
    \label{eqn:jet_flow_Galerkin_ROM}
\end{equation}
where $\tilde{f}$ was evaluated using pre-computed tensors assembled from the quadratic function $f_0$ and the columns of $\Phi$ and $\Psi$ as in Section~4.2 of \cite{Holmes2012turbulence}.
While wall-clock timing is implementation dependent, our simulations of \cref{eqn:jet_flow_Galerkin_ROM} took $\sim 0.5$ seconds per trajectory of length $100\Delta t$ on a laptop computer, whereas simulating \cref{eqn:jet_flow_FOM} took $\sim 180$ seconds per trajectory.

The POD projection was computed using the snapshots from the training trajectories.
We computed the BPOD projection using a time horizon equal to the length of each training trajectory and $20$-dimensional output projection.
For CoBRAS, the gradient was sampled along the training trajectories using the method described in \cref{subsec_long_trajectories} with a time horizon consisting of $L = 40$ intervals of length $\Delta t$ and $N=100$ (we simulated an extra $40\Delta t$ at the end of each training trajectory).
We solved the adjoint equation $10$ times per trajectory, yielding $s_g = 120$.
Per the discussion at the end of \cref{subsec_long_trajectories}, the computational cost to obtain the gradient samples was comparable to simulating the FOM \cref{eqn:jet_flow_FOM} over a time horizon equal to $s_g \cdot L \cdot \Delta t = 4800 \Delta t$. 

The CoBRAS projection was used to initialize the gradient descent in TrOOP with the same cost function for the jet flow described in \cite{Otto2022optimizing}.
This cost function is given by
\begin{equation}
    J(P) 
    = \frac{1}{12} \sum_{k=1}^{12} \frac{\sum_{l=0}^{99} \Vert x^{(k)}(l\Delta t) - \hat{x}^{(k)}(l\Delta t; P) \Vert^2}{\sum_{l=0}^{99} \Vert x^{(k)}(l\Delta t) \Vert^2}
    + \gamma \rho(P),
    \label{eqn:TrOOP_cost_fun}
\end{equation}
where $\hat{x}^{(k)}$ denotes the projection-dependent state forecast of \cref{eqn:jet_flow_Galerkin_ROM} corresponding to the $k$th training trajectory $x^{(k)}$.
The regularization function $\rho$ is defined in \cite{Otto2022optimizing} and we use the same regularization strength $\gamma = 10^{-3}$.
At each step of the geometric conjugate gradient method (Algorithm~4.2 in \cite{Otto2022optimizing}), line search was carried out in order to satisfy the weak Wolfe conditions with $c_1 = 0.01$ and $c_2 = 0.1$.
The gradient was computed using the adjoint sensitivity method described by Algorithm~4.1 in \cite{Otto2022optimizing} with $q=3$ Gauss-Legendre quadrature points per sampling interval.
The convergence of TrOOP with initialization provided by CoBRAS and POD are compared in \cref{fig:TrOOP_convergence}.
In both cases, we performed $500$ iterations.
We observe that after roughly $20$ iterations with POD initialization, the gradient rapidly decreases by thee orders of magnitude and the cost remains approximately constant,
indicating that TrOOP has become stuck near a local minimum of the cost function.
A much better initialization is provided by CoBRAS, which allows TrOOP to reach a lower value of the cost. 

\begin{figure}
    \centering
    \subfloat[projection error\label{fig:jet_Galerkin_rec_performance}]{
    \begin{tikzonimage}[trim=10 10 40 20, clip=true, width=0.450\textwidth]{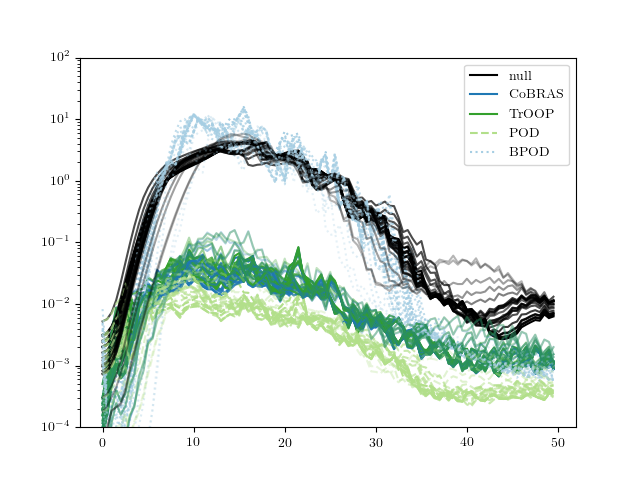}
    \node[rotate=90] at (0.0, 0.5) {\footnotesize $\Vert x - P x \Vert^2 / \avg{\big(\Vert x \Vert^2\big)}$};
    \node[rotate=0] at (0.55, 0.01) {\footnotesize Time $t$};
    \end{tikzonimage}
    }
    \subfloat[forecasting error \label{fig:jet_Galerkin_performance}]{
    \begin{tikzonimage}[trim=10 10 40 20, clip=true, width=0.450\textwidth]{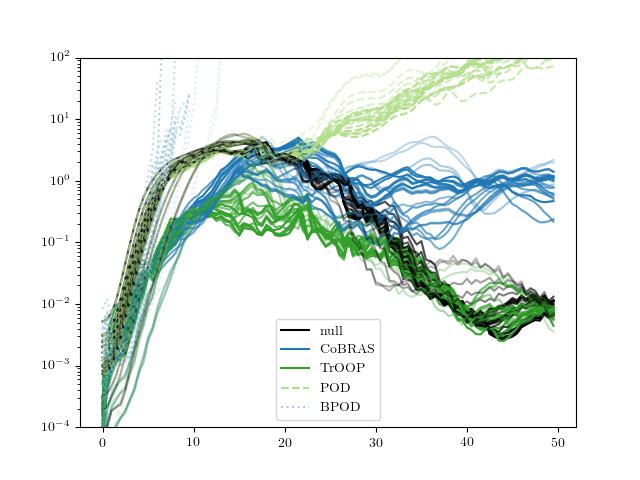}
    \node[rotate=90] at (0.0, 0.5) {\footnotesize $\Vert \hat{x} - x \Vert^2 / \avg{\big(\Vert x \Vert^2\big)}$};
    \node[rotate=0] at (0.55, 0.01) {\footnotesize Time $t$};
    \end{tikzonimage}
    }
    \caption{In (a) we show the reconstruction error using the different linear projections along each testing trajectory. The ``null'' projection means $Px = 0$. In (b) we show the error of the predictions made by the Petrov-Galerkin ROMs. The ``null'' forecast means $\hat{x} = 0$. In (a) and (b) we normalize the errors by the mean kinetic energy along each trajectory. The opacity of the trajectories increases with $\avg{\big(\Vert x \Vert^2\big)}$.}
    \label{fig:jet_impulse}
\end{figure}

The error introduced by each projection in reconstructing the states along the testing trajectories is plotted in \cref{fig:jet_Galerkin_rec_performance}.
Unsurprisingly, the POD projection has the lowest square error at nearly all times.
On the other hand, BPOD cannot accurately reconstruct states that depart from the responses of the linearized system.
Both CoBRAS and TrOOP have larger reconstruction error than POD, except at very early times.
Until roughly $t=5$, the reconstruction error for CoBRAS and TrOOP is lower than POD.

We compare the performance of the ROMs in terms of forecasting error normalized by the mean kinetic energy along each testing trajectory in \cref{fig:jet_Galerkin_performance}.
Predicted flowfield snapshots along the most energetic testing trajectory are shown in \cref{fig:Galerkin_snapshots}.
We observe that the POD-based model does not capture the growth of the disturbance, likely because the most energetic POD modes were primarily supported downstream.
On the other hand the BPOD-based model captures the initial growth, but rapidly blows up as nonlinearities become significant.
Between these extremes is CoBRAS, which captures the initial growth and has physically plausible, yet quantitatively inaccurate behavior at later times.
Though CoBRAS provided quantitatively accurate forecasts for only a short time, it produced a suitable initial projection for TrOOP, which was able to significantly reduce the error at later times.

Similar and slightly improved results were also obtained using CoBRAS with longer gradient sampling horizons $L$ and more gradient samples.
The performance of POD and BPOD was largely independent of the model dimension, while the dimension had a significant effect on the performance of CoBRAS.
Certain model dimensions consistently produced accurate predictions regardless of the gradient sampling strategy, while others produced models that blew up on some of the larger trajectories after correctly predicting the initial growth of the disturbance up to $t \approx 15$.

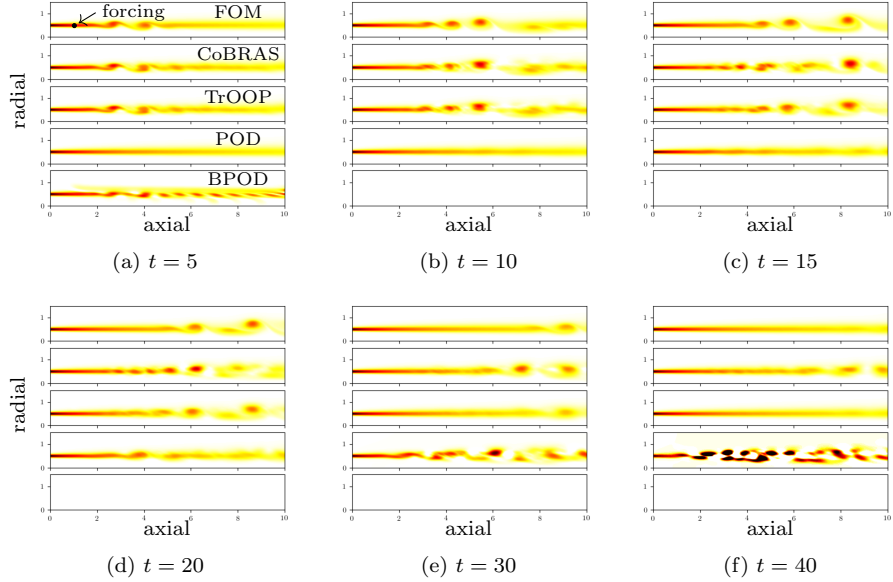
\begin{figure}
    \centering
    \subfloat[$t=5$\label{fig:Galerkin_snap_1}]{
    \begin{tikzonimage}[trim=50 15 40 40, clip=true, width=0.29\textwidth]{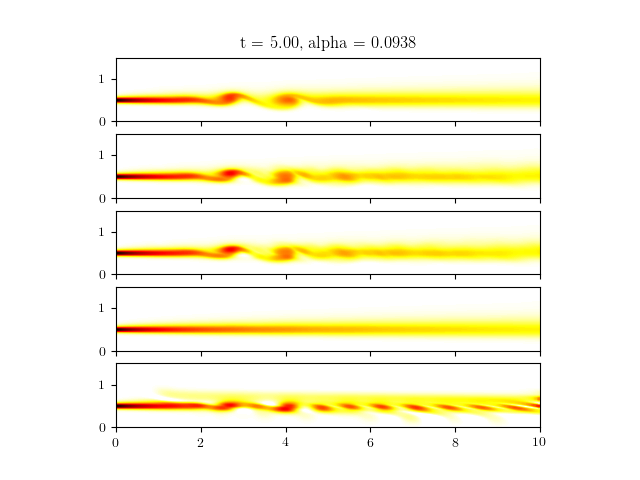}
    \node[rotate=0] at (0.5, 0.0) {\footnotesize axial};
    \node[rotate=90] at (-0.025,0.55) {\footnotesize radial};
    \node at (0.75,0.95) {\scriptsize FOM};
    \node at (0.75,0.76) {\scriptsize CoBRAS};
    \node at (0.75,0.57) {\scriptsize TrOOP};
    \node at (0.75,0.38) {\scriptsize POD};
    \node at (0.75,0.20) {\scriptsize BPOD};
    \filldraw[black] (0.175,0.887) circle (0.75pt);
    \node[anchor=west] at (0.24,0.945) (nodeB) {\scriptsize forcing};
    \draw [->] (0.25,0.945) -- (.19,0.898);
    \end{tikzonimage}
    }
    \subfloat[$t=10$\label{fig:Galerkin_snap_2}]{
    \begin{tikzonimage}[trim=50 15 40 40, clip=true, width=0.29\textwidth]{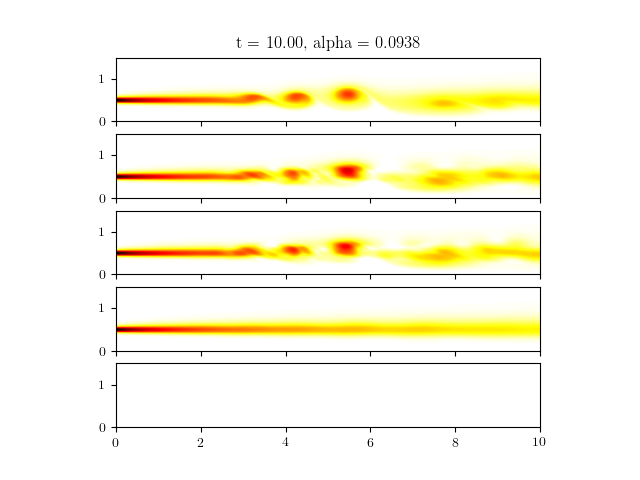}
    \node[rotate=0] at (0.5, 0.0) {\footnotesize axial};
    \end{tikzonimage}
    }
    \subfloat[$t=15$\label{fig:Galerkin_snap_3}]{
    \begin{tikzonimage}[trim=50 15 40 40, clip=true, width=0.29\textwidth]{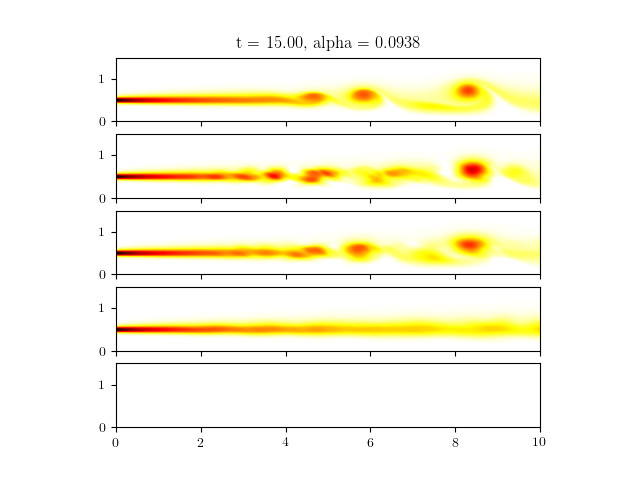}
    \node[rotate=0] at (0.5, 0.0) {\footnotesize axial};
    \end{tikzonimage}
    } \\
    \subfloat[$t=20$\label{fig:Galerkin_snap_4}]{
    \begin{tikzonimage}[trim=50 15 40 40, clip=true, width=0.29\textwidth]{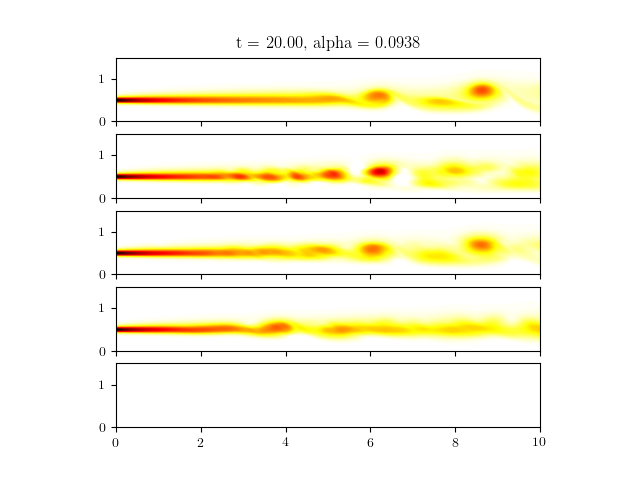}
    \node[rotate=0] at (0.5, 0.0) {\footnotesize axial};
    \node[rotate=90] at (-0.025,0.55) {\footnotesize radial};
    \end{tikzonimage}
    }
    \subfloat[$t=30$\label{fig:Galerkin_snap_5}]{
    \begin{tikzonimage}[trim=50 15 40 40, clip=true, width=0.29\textwidth]{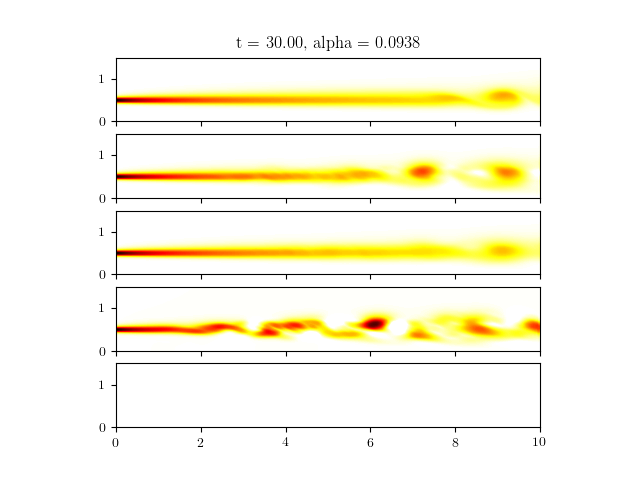}
    \node[rotate=0] at (0.5, 0.0) {\footnotesize axial};
    \end{tikzonimage}
    }
    \subfloat[$t=40$\label{fig:Galerkin_snap_6}]{
    \begin{tikzonimage}[trim=50 15 40 40, clip=true, width=0.29\textwidth]{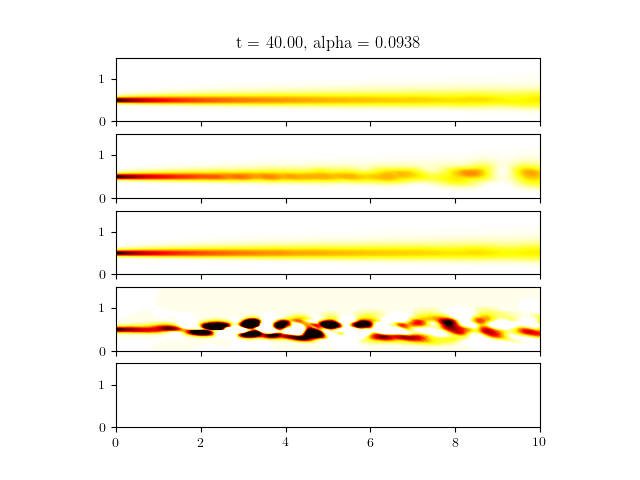}
    \node[rotate=0] at (0.5, 0.0) {\footnotesize axial};
    \end{tikzonimage}
    }
    \caption{Snapshots of predicted vorticity in the jet flow using
      40-dimensional Petrov-Galerkin ROMs along the most energetic testing trajectory, which had impulse magnitude $u_0 = 0.0938$.}
    \label{fig:Galerkin_snapshots}
\end{figure}

\subsubsection{Learned models in nonlinear feature space}
In this section we compare $15$ and $30$-dimensional ROMs for the jet flow in the spaces of nonlinear coordinates found by KPCA and the kernel CoBRAS (K-CoBRAS) method described in \cref{subsec:kernel_CoBRAS}.
For both methods we use the Gaussian kernel with width $\sigma = 8.0$ (see the third row in \cref{tab:common_kernels}).
The same gradient samples described above in \cref{subsubsec:jet_Galerkin} were used to construct the K-CoBRAS embedding. 
To preserve the location of the equilibrium, the KPCA features were centered about the origin by projecting onto the leading eigenvectors of the empirical state covariance operator \cref{eqn:lifted_state_covariance}.

The training trajectories described at the beginning of \cref{sec:jet_flow} are plotted in the resulting nonlinear feature spaces $z = h(x)$ using the leading three KPCA coordinates in \cref{fig:KPOD_embedding} and the leading three K-CoBRAS coordinates in \cref{fig:KCoBRAS_embedding}.
To make these plots, we interpolated the embedded data using piece-wise cubic polynomials and derivative information $\dot{z} = \D h(x) \dot{x}$ obtained at each sample point using the kernel tangent map \cref{eqn:kernel_embedding_tangent_map} and the FOM.
While the KPCA coordinates reflect the energetic growth and decay along trajectories, the points at early and late times are mapped very closely to the equilibrium point at the origin.
On the other hand, K-CoBRAS relies on sensitivity information and consequently maps the initial conditions to points far away from the origin.
The resulting trajectories spiral inward towards the stable equilibrium point.
The leading K-CoBRAS coordinates primarily capture the initial growth of the disturbance in an upstream region, with at least $r=10$ coordinates being necessary to represent the behavior of downstream vortices.

\begin{figure}
    \centering
    \subfloat[KPCA embedding\label{fig:KPOD_embedding}]{
    \begin{tikzonimage}[trim=110 30 65 60, clip=true, width=0.35\textwidth]{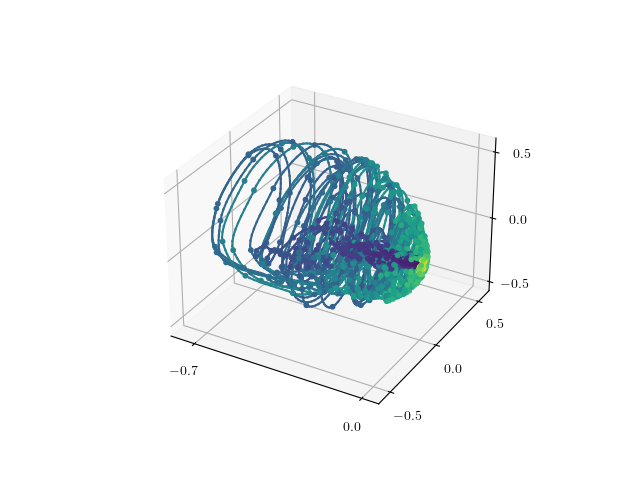}
        \node[rotate=0] at (0.25, 0.10) {\footnotesize $z_1$};
        \node[rotate=0] at (0.85, 0.20) {\footnotesize $z_2$};
        \node[rotate=0] at (1.00, 0.60) {\footnotesize $z_3$};
    \end{tikzonimage}
    }
    \hspace{0.1\textwidth}
    \subfloat[K-CoBRAS embedding\label{fig:KCoBRAS_embedding}]{
    \begin{tikzonimage}[trim=110 30 65 60, clip=true, width=0.35\textwidth]{Figures/jet_Re2000/train_trajs_KCoBRAS_embedding.png}
        \node[rotate=0] at (0.09, 0.19) {\footnotesize $z_1$};
        \node[rotate=0] at (0.71, 0.09) {\footnotesize $z_2$};
        \node[rotate=0] at (1.00, 0.50) {\footnotesize $z_3$};
    \end{tikzonimage}
    }
    \caption{In (a) and (b) we show the training trajectories of the jet flow in the leading three nonlinear coordinates found by KPCA and K-CoBRAS respectively. The coloring is from blue to yellow according to time.}
    \label{fig:kernel_embedding}
\end{figure}
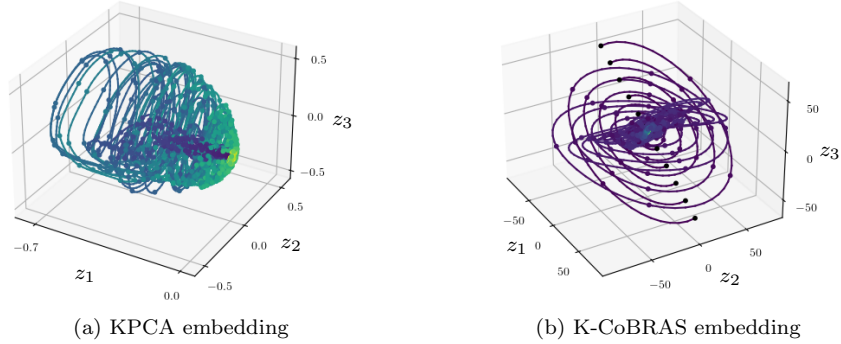

To construct ROMs using the nonlinear coordinates $z = h(x)$ extracted by KPCA and K-CoBRAS, we employed kernel ridge regression (KRR) to fit the discrete-time dynamics $\tilde{f}$ and reconstruction map $\tilde{g}$ to the training data described at the beginning of \cref{sec:jet_flow}.
This yields a model in the form of \cref{eqn:ROM} where $\tilde{f}$ approximates the map $z(t) \mapsto z(t + \Delta t)$ and $\tilde{g}$ approximates the reconstruction map $z \mapsto x$.
Since it would be impractical to fit a reconstruction function for each of the $10^5$ state variables in the jet flow, we used the (nonlinear) KPCA and K-CoBRAS coordinates to reconstruct the leading $100$ (linear) POD and CoBRAS coordinates respectively.
The full state was then linearly reconstructed in the $100$-dimensional POD and CoBRAS bases.
Gaussian radial basis function (RBF) kernels were used for regression with the parameters listed in \cref{tab:KRR_params} of \cref{app:kernel_params} selected by $5$-fold cross-validation over parameter grids.
The fitting process was carried out using the ``KernelRidge'' and ``GridSearchCV'' tools in Scikit-Learn \cite{scikit-learn} after normalizing the variance of each coordinate.
Though wall-clock timing is implementation dependent, our forecasts took $\sim 2$ seconds per trajectory using the $15$ and $30$-dimensional K-CoBRAS models on a laptop computer.
Forecasts using the $15$ and $30$-dimensional KPCA models took $\sim 5$ and $\sim 13$ seconds per trajectory, respectively.

The prediction accuracy of the resulting data-driven ROMs on the testing trajectories is shown in \cref{fig:kernel_performance}.
The models built in K-CoBRAS coordinates have superior prediction accuracy at nearly all times, with accuracy increasing with the number of coordinates in the model.
On the other hand, the models in KPCA coordinates are no better than the null forecast, though their reconstruction error is lower than the K-CoBRAS models after roughly $t=10$.
The K-CoBRAS reconstructions are more accurate at early times, allowing these models to capture the initial growth of disturbances in the upstream region of the flow.
Examining the predicted snapshots in \cref{fig:kernel_KR_snapshots} using the $15$-dimensional models along the most energetic testing trajectory 
confirms that the model in KPCA coordinates fails to capture the upstream disturbances, yielding inaccurate forecats.
On the other hand, the fitted models in K-CoBRAS coordinates accurately predict the flow's response even at late times.

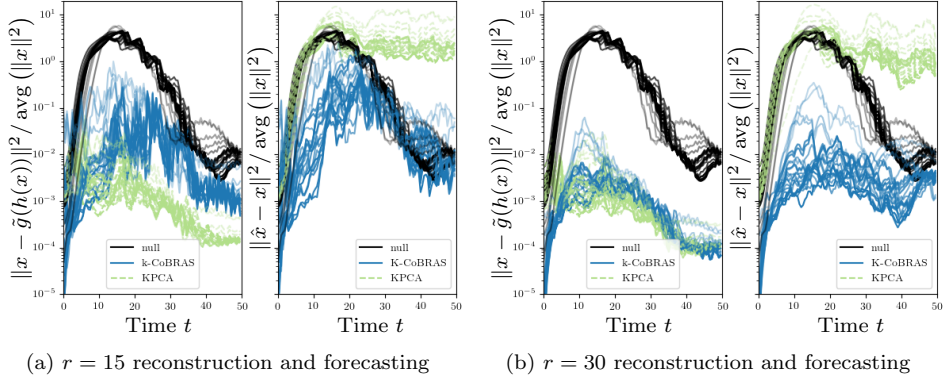
\begin{figure}
    \centering
    \subfloat[$r=15$ reconstruction and forecasting \label{fig:kernel_dim15_error}]{
    \begin{tikzonimage}[trim=20 10 40 20, clip=true, width=0.45\textwidth]{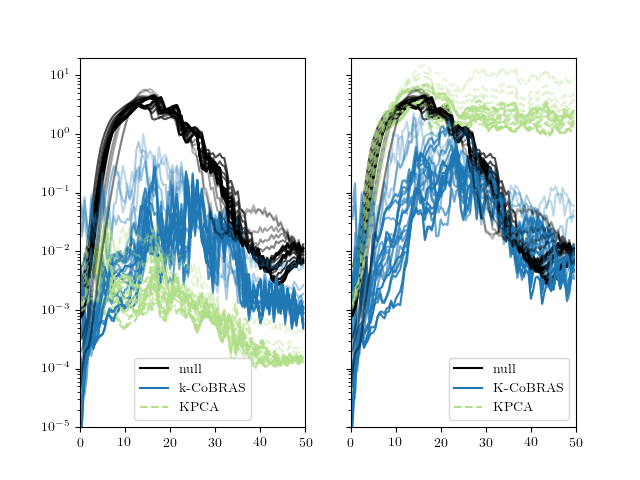}
        \node[rotate=90] at (0.0, 0.510) {\footnotesize $\Vert x - \tilde{g}(h(x)) \Vert^2 / \avg{\big(\Vert x \Vert^2\big)}$};
        \node[rotate=0] at (0.31, 0.00) {\footnotesize Time $t$};
        \node[rotate=90] at (0.54, 0.510) {\footnotesize $\Vert \hat{x} - x \Vert^2 / \avg{\big(\Vert x \Vert^2\big)}$};
        \node[rotate=0] at (0.79, 0.00) {\footnotesize Time $t$};
    \end{tikzonimage}
    }
    \subfloat[$r=30$ reconstruction and forecasting \label{fig:kernel_dim30_error}]{
    \begin{tikzonimage}[trim=20 10 40 20, clip=true, width=0.45\textwidth]{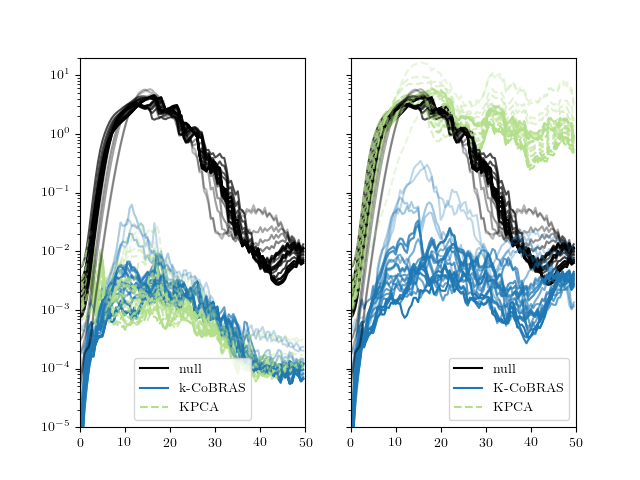}
        \node[rotate=90] at (0.0, 0.510) {\footnotesize $\Vert x - \tilde{g}(h(x)) \Vert^2 / \avg{\big(\Vert x \Vert^2\big)}$};
        \node[rotate=0] at (0.31, 0.00) {\footnotesize Time $t$};
        \node[rotate=90] at (0.54, 0.510) {\footnotesize $\Vert \hat{x} - x \Vert^2 / \avg{\big(\Vert x \Vert^2\big)}$};
        \node[rotate=0] at (0.79, 0.00) {\footnotesize Time $t$};
    \end{tikzonimage}
    }
    \caption{We show the reconstruction and forecasting error along testing trajectories using the $15$ and $30$-dimensional models of the jet flow constructed in KPCA and K-CoBRAS coordinates. For the ``null'' predictions, we reconstruct/forecast zero. Opacity of the curves increases with $\avg{\big(\Vert x \Vert^2\big)}$.}
    \label{fig:kernel_performance}
\end{figure}

\begin{figure}
    \centering
    \subfloat[$t=5$\label{fig:kernel_KR_snap_1}]{
    \begin{tikzonimage}[trim=50 20 40 50, clip=true, width=0.29\textwidth]{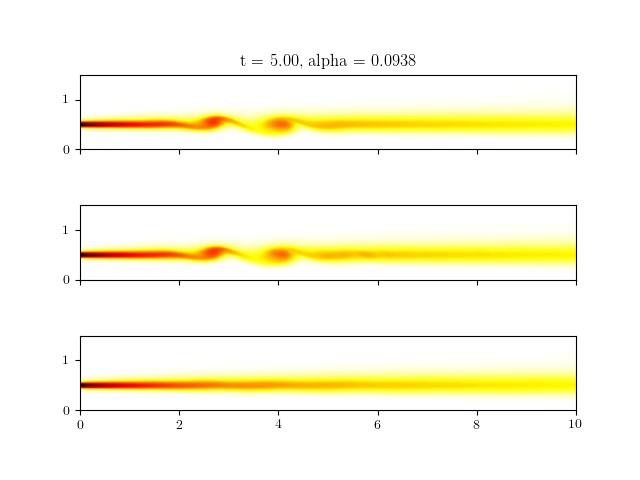}
    \node[rotate=0] at (0.5, 0.0) {\footnotesize axial};
    \node[rotate=90] at (-0.05,0.54) {\footnotesize radial};
    \node at (0.75,0.93) {\scriptsize FOM};
    \node at (0.75,0.59) {\scriptsize K-CoBRAS};
    \node at (0.75,0.25) {\scriptsize KPCA};
    \end{tikzonimage}
    }
    \subfloat[$t=10$\label{fig:kernel_KR_snap_2}]{
    \begin{tikzonimage}[trim=50 20 40 50, clip=true, width=0.29\textwidth]{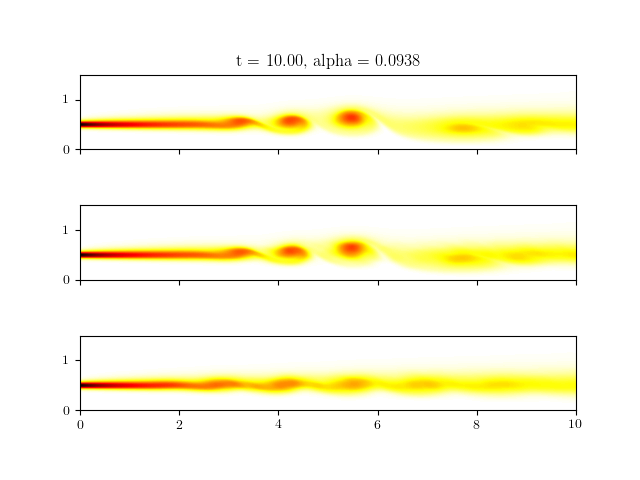}
    \node[rotate=0] at (0.5, 0.0) {\footnotesize axial};
    \end{tikzonimage}
    }
    \subfloat[$t=15$\label{fig:kernel_KR_snap_3}]{
    \begin{tikzonimage}[trim=50 20 40 50, clip=true, width=0.29\textwidth]{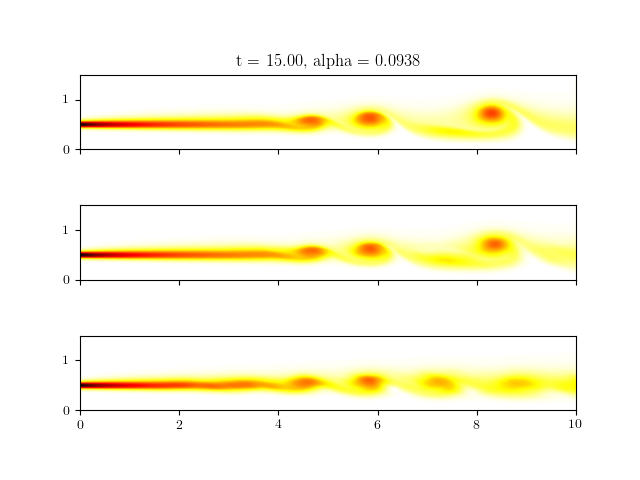}
    \node[rotate=0] at (0.5, 0.0) {\footnotesize axial};
    \end{tikzonimage}
    } \\
    \subfloat[$t=20$\label{fig:kernel_KR_snap_4}]{
    \begin{tikzonimage}[trim=50 20 40 50, clip=true, width=0.29\textwidth]{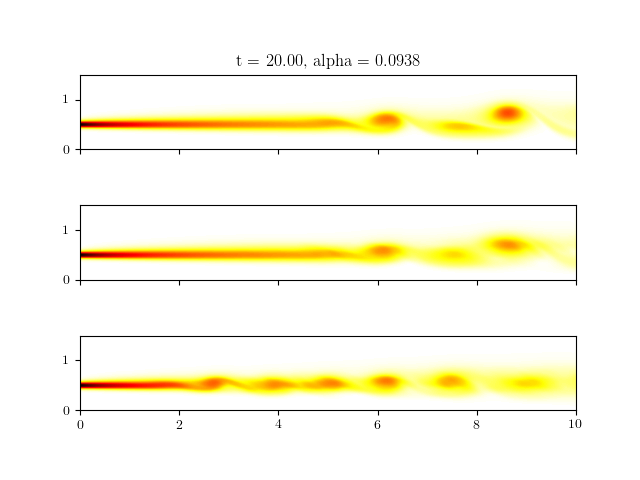}
    \node[rotate=0] at (0.5, 0.0) {\footnotesize axial};
    \node[rotate=90] at (-0.05,0.54) {\footnotesize radial};
    \end{tikzonimage}
    }
    \subfloat[$t=30$\label{fig:kernel_KR_snap_5}]{
    \begin{tikzonimage}[trim=50 20 40 50, clip=true, width=0.29\textwidth]{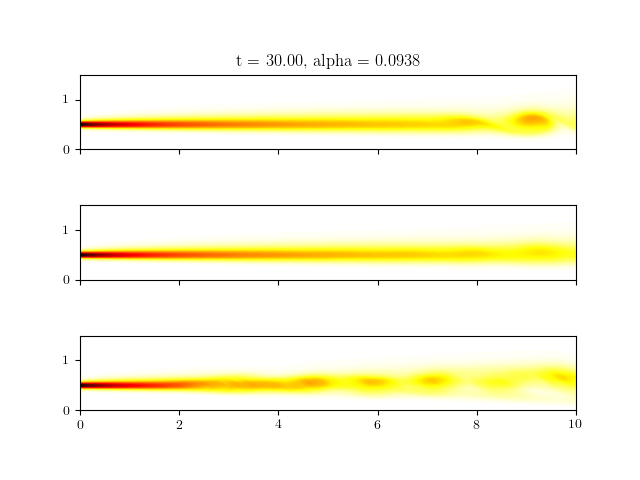}
    \node[rotate=0] at (0.5, 0.0) {\footnotesize axial};
    \end{tikzonimage}
    }
    \subfloat[$t=40$\label{fig:kernel_KR_snap_6}]{
    \begin{tikzonimage}[trim=50 20 40 50, clip=true, width=0.29\textwidth]{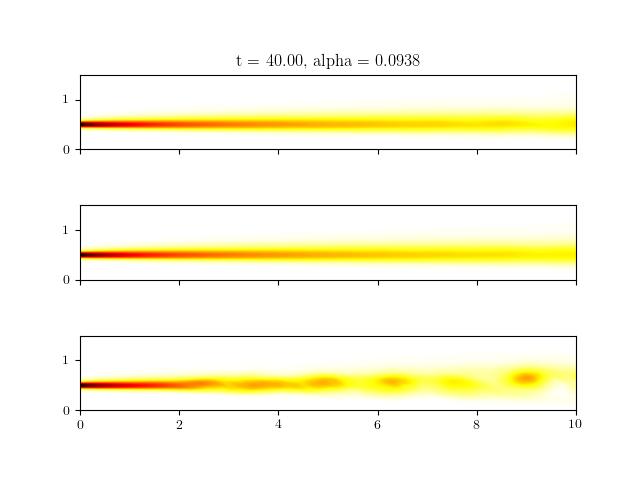}
    \node[rotate=0] at (0.5, 0.0) {\footnotesize axial};
    \end{tikzonimage}
    }
    \caption{Snapshots of predicted vorticity in the jet flow using learned
      15-dimensional models in kernel feature space along the most energetic testing trajectory, which had impulse magnitude $u_0 = 0.0938$.}
    \label{fig:kernel_KR_snapshots}
\end{figure}
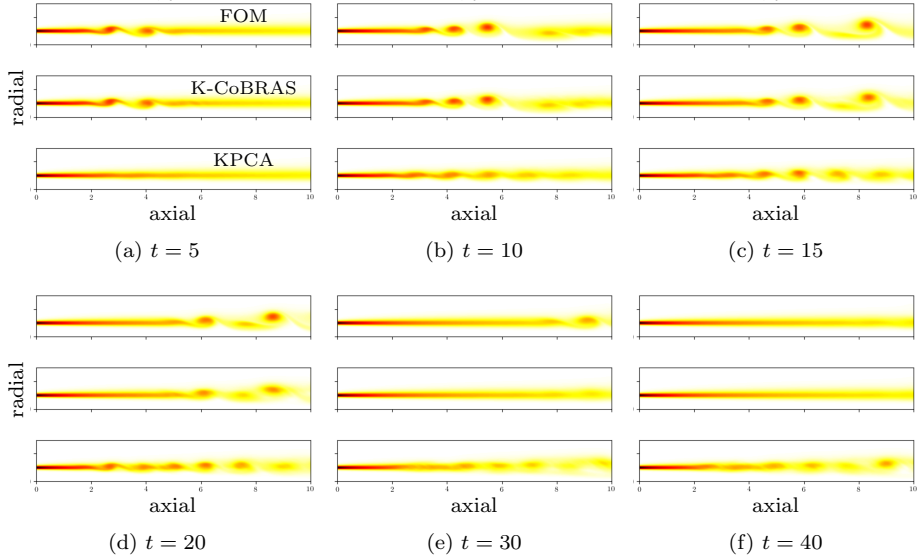

\section{Conclusion}

We have introduced covariance balancing reduction using adjoint snapshots (CoBRAS) as a method for finding low-dimensional coordinates for model reduction of nonlinear dynamical systems.
These coordinates capture information about future outputs of the system in the sense of \cite{Zahm2020gradient} by balancing the sensitivity of future outputs against the variance of the distribution of states generated along trajectories, as measured by empirical state and gradient covariance matrices.
The resulting method is analogous to balanced truncation (BT) \cite{Moore1981principal} with covariance matrices replacing the system Gramians and obeying the same transformation laws.
An efficient snapshot-based computational procedure is provided by analogy to balanced proper orthogonal decomposition (BPOD) \cite{Rowley2005model}.
The features extracted by CoBRAS are associated with an oblique projection operator that can be used directly for constructing Petrov-Galerkin models.
We demonstrate the performance of CoBRAS-Galerkin models on a simple, yet challenging three-dimensional system as well as on a jet flow simulation with $10^5$ state variables.
Not only did the CoBRAS-Galerkin models perform well, but they also provided superior initializations compared to POD for the recently proposed method of trajectory-based optimization for oblique projections (TrOOP) \cite{Otto2022optimizing}.

We observe that the CoBRAS features (coordinates) depend only on inner products between state and gradient vectors.
This allows us to replace the inner product with a kernel function in order to identify nonlinear features in a higher-dimensional reproducing kernel Hilbert space (RKHS) using a method we refer to as kernel-CoBRAS (K-CoBRAS).
While the associated projection is now defined in the RKHS and cannot be used for Petrov-Galerkin projection of the full-order model, it is possible to construct ROMs governing the evolution of the extracted nonlinear coordinates by employing regression-based approaches.
We applied this approach to the jet flow and compared the performance of the K-CoBRAS coordinates to kernel principal component analysis (KPCA) coordinates by fitting the discrete-time dynamics and reconstruction maps via kernel ridge regression.
The learned K-CoBRAS model accurately captured the dynamics, while the KPCA model failed to capture the growth of disturbances.
We argue that this likely because KPCA truncated coordinates associated with the low-variance upstream growth, while K-CoBRAS retained these features due to its reliance on sensitivity information.

Our code was written in Python and is available at \url{https://github.com/samotto1/CoBRAS}.


\bibliographystyle{siamplain}
\bibliography{jfull,references}


\appendix

\section{Proof of Theorem~\ref{thm:projection_error_for_Gaussian}}
\label{app:projection_error_for_Gaussian}

By the tower rule for conditional expectation, we have
\begin{equation}
    \E\left[ \big\Vert F(x, \bar{u}) - \hat{F}(P x, \bar{u}) \big\Vert^2 \right]
    = \E \left[ \E\left[ \big\Vert F(x, \bar{u}) - \hat{F}(P x, \bar{u}) \big\Vert^2 \ \big\vert \ \bar{u} \right] \right].
\end{equation}
Letting $W_{g\vert\bar{u}} = \E\big[ \grad_x F(x, \bar{u}) \grad_x F(x, \bar{u})^T \ \vert \ \bar{u} \big]$ and applying Proposition~2.5 in Zahm et al. \cite{Zahm2020gradient} to the conditional expectation yields
\begin{equation}
    \E\left[ \big\Vert F(x, \bar{u}) - \hat{F}(P x, \bar{u}) \big\Vert^2 \right]
    \leq \E\left[ \Tr\left[ W_{g\vert\bar{u}} (I-P) \Sigma_{x\vert\bar{u}} (I-P)^T \right] \right].
\end{equation}
Using our assumption about $C$, it follows that
\begin{equation}
    C W_x - \Sigma_{x\vert\bar{u}}
    = C \underbrace{(W_x - \Sigma_x)}_{\E[x] \E[x]^T} + C \Sigma_x - \Sigma_{x\vert\bar{u}}
    \succeq 0
\end{equation}
is positive semi-definite almost surely.
Applying this result, linearity of the expectation, and the tower rule, we obtain
\begin{align}
    \E\left[ \big\Vert F(x, \bar{u}) - \hat{F}(P x, \bar{u}) \big\Vert^2 \right] 
    \leq & \E\left[ \Tr\left[ W_{g\vert\bar{u}} (I-P) C W_x (I-P)^T \right] \right] \\
    & = C \Tr\left[ \E \left[ W_{g\vert\bar{u}}\right] (I-P) W_x (I-P)^T \right] \\
    & = C \Tr\left[ W_g (I-P) W_x (I-P)^T \right].
\end{align}

It remains to verify that we can take $C=1$ in the jointly Gaussian case.
Thanks to Proposition~3.13 in Eaton \cite{Eaton2007multivariate},
when $(x, \bar{u})$ has Gaussian distribution, the conditional covariance matrix is given explicitly by
\begin{equation}
    \Sigma_{x\vert\bar{u}} 
    = \Sigma_x - \Sigma_{x,\bar{u}} \Sigma_{\bar{u}}^{+} \Sigma_{x,\bar{u}}^T,
\end{equation}
where $\Sigma_{x,\bar{u}} = \E\big[ (x- \E[x]) (\bar{u} - \E[\bar{u}])^T\big]$ and $\Sigma_{\bar{u}}^{+}$ denotes the Moore-Penrose pseudoinverse of $\Sigma_{\bar{u}} = \E\big[ (\bar{u} - \E[\bar{u}]) (\bar{u} - \E[\bar{u}])^T\big]$.
Since $\Sigma_x - \Sigma_{x\vert\bar{u}} = \Sigma_{x,\bar{u}} \Sigma_{\bar{u}}^{+} \Sigma_{\bar{u},x} \succeq 0$ is positive semi-definite, we can take $C=1$.
This completes the proof of the theorem.

\section{Proof of Theorem~\ref{thm:classification_of_data_equivalent_subspaces}}
\label{app:classification_of_data_equivalent_subspaces}

We begin by proving \cref{lem:generic_sampling}, which ensures that almost every matrix $X$ with respect to Lebesgue measure satisfes the hypotheses of the theorem.
\begin{proof}[Proof of \cref{lem:generic_sampling}]
Let $\Phi \in \R^{n\times r}$ have linearly independent columns spanning $\Range(M^T)$ and let $A \in \R^{s\times r}$ be a matrix with linearly independent columns.
The functions $\phi: X \mapsto \det(X^T X)$ and $\psi: X\mapsto \det(\Phi^T X A)$ are non-constant polynomials on the Euclidean space $\R^{n\times s}$.
If follows from the main result in \cite{Caron2005zero} that the level sets of $\phi$ and $\psi$ have zero Lebesgue measure.
The level set $\phi^{-1}(0)$ corresponds to matrices $X \in \R^{n\times s}$ whose columns are linearly dependent.
The level set $\psi^{-1}(0)$ contains every matrix $X \in \R^{n\times s}$ for which $\rank(M X) < r$.
Thus, the union $\phi^{-1}(0) \cup \psi^{-1}(0)$ has Lebesgue measure zero and $\phi^{-1}(0) \cup \psi^{-1}(0)$ contains every $X \in \R^{n\times s}$ failing to have linearly independent columns or failing to satisfy $\rank(M X) = r$.
\end{proof}

We introduce some notation and recall some basic concepts that we use throughout our proof of Theorem~\ref{thm:classification_of_data_equivalent_subspaces}.
If $M$ is a matrix, let $\row_{i_0 : i_1}(M)$ we denote the submatrix of $M$ formed from rows $i_0$ through $i_1$ with $i_1 \geq i_0$.
Similarly, $\col_{j_0:j_1}$ denotes the submatrix formed from columns $j_0$ through $i_1$.
The smooth manifold of $q\times p$ matrices with linearly independent columns is denoted $\R_*^{q\times p}$.
We recall that the canonical projection map $\pi_{q,p}:\R_*^{q\times p} \to \calG_{q,p}$ given by $\pi(\Phi) = \Range(\Phi)$ is a surjective submersion.
This follows from the quotient manifold theorem (Theorem 21.10 in Lee \cite{Lee2013introduction}) and the fact that $\calG_{q,p}$ is the quotient of $\R_*^{q\times p}$ under a free and proper action of the general linear group of invertible $p\times p$ matrices.
Thanks to Theorem~4.29 in \cite{Lee2013introduction}, a function $f$ on $\calG_{q,p}$ is smooth if and only if $f\circ\pi_{q,p}$ is smooth.
We also make use of the following method for constructing smooth functions on the Grassmannian.
Suppose that $\bar{f}$ is a smooth function on $\R_*^{q\times p}$ that is constant on the fibers of $\pi_{q,p}$, i.e., $\bar{f}$ satisfies $\bar{f}(\Psi) = \bar{f}(\Psi A)$ for every $\Psi\in\R_*^{q\times p}$ and invertible $p\times p$ matrix $A$.
Thanks to Theorem~4.30 in \cite{Lee2013introduction} there is a unique smooth function $f$ on $\calG_{q,p}$ satisfying $\bar{f} = f \circ \pi_{q,p}$, that is, $\bar{f}(\Psi) = f(\Range(\Psi))$ for every $\Psi\in\R_*^{q\times p}$.

\begin{lemma}[Grassmannian cross-section, see Absil et al. \cite{Absil2004riemannian}]
    \label{lem:cross_section_lemma}
    Let $1 \leq p \leq q$ be integers.
    For each $M\in \R_*^{q\times p}$, the affine ``cross-section''
    \begin{equation}
        S_M = \left\{ Y \in \R_*^{q\times p} \ : \ M^T(Y - M) = 0 \right\}
    \end{equation}
    is mapped diffeomorphically onto the open subset of the Grassmannian
    \begin{equation}
        \calU_M = \left\{ \Range(Y)\in\calG_{q,p} \ : \ \det(M^T Y) \neq 0 \right\}
    \end{equation}
    by the canonical projection $\pi_{q,p}:Y \mapsto \Range(Y)$.
    The smooth inverse is the ``cross-section mapping'' $\sigma_M: \calU_M \to S_M$ defined by
    \begin{equation}
        \sigma_M : \Range(Y) \mapsto Y (M^T Y)^{-1} M^T M.
    \end{equation}
\end{lemma}
\begin{proof}
    See Absil et al. \cite{Absil2004riemannian}.
\end{proof}
Setting $M = \begin{bmatrix} I_p & 0 \end{bmatrix}^T$ in the above lemma yields the following useful corollary
\begin{corollary}
    \label{cor:Wqp}
    Let $1 \leq p \leq q$ be integers.
    The subset
    \begin{equation}
        \calW_{q,p} = 
        \left\{ \Range(Y)\in\calG_{q,p} \ : \ \det\big(\row_{1:p}(Y)\big) \neq 0 \right\}
    \end{equation}
    is open in the Grassmannian and diffeomorphic to $\R^{(q-p)\times p}$.
    In particular, a diffeomorphism $\phi_{q,p}: \R^{(q-p)\times p} \to \calW_{q,p}$ is given by
    \begin{equation}
        \phi_{q,p}: A \mapsto 
        \Range\left(
        \begin{bmatrix}
            I_p \\
            A
        \end{bmatrix}
        \right).
    \end{equation}
\end{corollary}
\begin{proof}
    Taking $M = \begin{bmatrix} I_p & 0 \end{bmatrix}^T$ in Lemma~\ref{lem:cross_section_lemma}, we observe that the canonical projection $\pi_{q,p}$ maps the cross-section
    \begin{equation}
        S_M = \left\{ Y \in \R_*^{q\times p} \ : \ \row_{1:p}(Y) = I_p \right\}
    \end{equation}
    diffeomorphically onto $\calW_{q,p}$.
    The proof is completed by observing that the cross-section is parametrized by the entries in the submatrix $\row_{(p+1):q}(Y)$.
\end{proof}

To prove Theorem~\ref{thm:classification_of_data_equivalent_subspaces}, we begin by reducing to the case where the columns of $X$ are given by the first $s$ columns of the $n\times n$ identity matrix, denoted $E_{n,s} = \begin{bmatrix} I_s & 0 \end{bmatrix}^T$.
Let $X = Q_X R$ be a reduced QR factorization where $Q_X$ is an $n\times s$ matrix whose columns are an orthonormal basis for the range of $X$ and $R$ is an invertible $s\times s$ matrix.
It follows that $\tilde{F}X = F X$ if and only if $\tilde{F} Q_X = F Q_X$.
Letting $Q$ be an $n\times n$ orthonormal matrix whose first $s$ columns are given by $Q_X$, we have $Q^T Q_X = E_{n,s}$.
It follows that $\tilde{F} Q_X = F Q_X$ if and only if $\tilde{F} Q E_{n,s} = F Q E_{n,s}$.
Since $\Range(\tilde{F}^T) = Q\Range\big((\tilde{F} Q)^T\big)$, we have
\begin{equation}
    \calV_{F,X} = Q \calV_{FQ, E_{n,s}},
\end{equation}
and so it suffices to study $\calV_{FQ, E_{n,s}}$.

We make use of the following characterization of sets $\calV_{F,X}$.
\begin{lemma}
    \label{lem:another_expression_for_V_X}
    Let $F$ and $X$ satisfy the hypotheses of \cref{thm:classification_of_data_equivalent_subspaces}.
    Then
    \begin{equation}
        \calV_{F,X} = \left\{ \Range(\Phi) \ : \ 
        \Phi \in \R^{n\times r} 
        \ \mbox{and} \ 
        \Range\big(X^T \Phi\big) = \Range\big(X^T F^T\big) \right\}.
    \label{eqn:set_of_equivalent_subspaces}
    \end{equation}
\end{lemma}
\begin{proof}
    Suppose that there is a matrix $\tilde{F}\in \R^{m\times n}$ with rank $r$ satisfying $\tilde{F}X = FX$. 
    Letting $\tilde{F} = U \Sigma \Phi^T$ be a reduced SVD, we have $\Range(\Phi) = \Range(\tilde{F})$.
    Since $X^T \Phi \Sigma U^T = X^T F^T$, we obtain $\Range\big(X^T F^T\big) \subset \Range\big(X^T \Phi\big)$.
    Since $\Sigma U^T$ is surjective, we conclude that $\Range\big(X^T \Phi\big) = \Range\big(X^T F^T\big)$.

    Conversely, suppose that $\Phi \in \R^{n\times r}$ satisfies $\Range\big(X^T \Phi\big) = \Range\big(X^T F^T\big)$.
    Then there is an $r\times m$ matrix $A$ so that
    \begin{equation}
        X^T \Phi A = X^T F^T.
    \end{equation}
    Since $\rank(F X) = r$, we must also have $\rank(\Phi A) = r$.
    Setting $\tilde{F} = A^T \Phi^T$ we conclude that $\Range(\Phi) = \Range(\tilde{F}^T)$ and $\tilde{F}X = F X$.
\end{proof}
Let $S_0$ be an $s\times s$ permutation matrix so that the first $r$ rows of $S_0^T E_{n,s}^T Q^T F^T$ are linearly independent.
Applying the Lemma~\ref{lem:another_expression_for_V_X} to $\calV_{FQ, E_{n,s}}$, we obtain
\begin{equation}
\begin{aligned}
    \calV_{FQ, E_{n,s}} = 
    \Big\{ \Range(\Phi) \ : \ 
        &\Phi \in \R^{n\times r} 
        \ \mbox{and} \ \\
        &\Range\big(S_0^T E_{n,s}^T \Phi\big) = \Range\big(S_0^T E_{n,s}^T Q^T F^T \big) \Big\}.
\end{aligned}
\end{equation}
We observe that the $n\times n$ permutation matrix
\begin{equation}
    S = \begin{bmatrix}
        S_0 & 0 \\
        0 & I_{n-s}
    \end{bmatrix},
\end{equation}
satisfies $S E_{n,s} = E_{n,s} S_0$, and so
\begin{equation}
\begin{aligned}
    \calV_{FQ, E_{n,s}} = 
    \Big\{ S\Range(S^T\Phi) \ : \ 
        &\Phi \in \R^{n\times r} 
        \ \mbox{and} \ \\
        &\Range\big(E_{n,s}^T S^T \Phi\big) = \Range\big(E_{n,s}^T S^T Q^T F^T \big) \Big\}.
\end{aligned}
\end{equation}
From this obtain $\calV_{FQ, E_{n,s}} = S \calV_{FQS, E_{n,s}}$, further reducing our problem to studying the set $\calV_{FQS, E_{n,s}}$.
To summarize the results so far, we have
\begin{equation}
    \calV_{F,X} 
    = Q \calV_{FQ, E_{n,s}}
    = Q S \calV_{FQS, E_{n,s}},
\end{equation}
where $Q S$ is a unitary transformation.

Let $M \in \R^{s\times r}$ be a matrix with linearly independent columns spanning the $r$-dimensional subspace $\Range\big(E_{n,s}^T S^T Q^T F^T \big)\subset \R^s$.
The set we aim to study can be expressed as
\begin{equation}
    \calV_{FQS, E_{n,s}}
    = \Big\{ \Range(\Phi) \ : \ 
        \Phi \in \R^{n\times r} 
        \ \mbox{and} \ 
        \Range\big( \row_{1:s}(\Phi) \big) = \Range\big( M \big) \Big\}.
    \label{eqn:V_reduced}
\end{equation}
By construction of $S_0$, the submatrix $\row_{1:r}(M)$ is invertible. 
This implies that $\row_{1:r}(\Phi)$ is invertible for every $\Phi\in\R^{n\times r}$ with $\Range(\Phi) = \Range(M)$.
Hence, we have
\begin{equation}
    \calV_{FQS, E_{n,s}} \subset \calW_{n,r},
    \label{eqn:VcontW}
\end{equation}
where $\calW_{n,r}$ is the open submanifold of $\calG_{n,r}$ described in Corollary~\ref{cor:Wqp}.

We construct a diffeomorphism $\calW_{n-s+r,r} \to \calV_{FQS, E_{n,s}}$ by first introducing the map $\bar{f}:\R_*^{(n-s+r)\times r} \to \R_*^{n\times r}$ defined by
\begin{equation}
    \bar{f}:\Psi \mapsto 
    \begin{bmatrix}
        \row_{1:r}(\Psi) \\
        \row_{r+1:s}(M) \row_{1:r}(M)^{-1} \row_{1:r}(\Psi) \\
        \row_{r+1:n-s+r}(\Psi)
    \end{bmatrix}.
    \label{eqn:V_map_lift}
\end{equation}
This map is well-defined because the matrix appearing on the right has $r$ columns and contains all of the rows of $\Psi$, of which $r$ are linearly independent.
Since this is a linear map of matrices, it is obviously smooth and gives rise to a smooth map $\pi_{n,r}\circ \bar{f}:\R_*^{(n-s+r)\times r} \to \calG_{n,r}$.
We observe that $\Range(\bar{f}(\Psi))$ only depends on the range of $\Psi$, for if $A$ is an invertible $r\times r$ matrix, then we have
\begin{equation}
    \Range\big(\bar{f}(\Psi A)\big) = 
    \Range\big(\bar{f}(\Psi) A\big)
    = \Range\big(\bar{f}(\Psi)\big).
\end{equation}
It follows from Theorem~4.30 in \cite{Lee2013introduction} that there is a unique smooth map $f:\calG_{n-s+r, r} \to \calG_{n, r}$ satisfying $\pi_{n,r}\circ\bar{f} = f \circ \pi_{n-s+r,r}$, that is,
\begin{equation}
    f(\Range(\Psi)) = \Range(\bar{f}(\Psi))
\end{equation}
for every $\Psi\in\R_*^{(n-s+r)\times r}$.
Restricting $f$ to the subset $\calW_{n-s+r,r} \subset \calG_{n-s+r, r}$ will provide the desired diffeomorphism onto $\calV_{FQS, E_{n,s}}$ as we now show.

First, we show that $f$ is injective on $\calG_{n-s+r, r}$.
If $f(\Range(\Psi_0)) = f(\Range(\Psi_1))$,
then we have $\Range(\bar{f}(\Psi_0)) = \Range(\bar{f}(\Psi_1))$.
This implies that there is an invertible $r\times r$ matrix $A$ so that $\bar{f}(\Psi_0) = \bar{f}(\Psi_1) A$.
Taking a subset of rows in Eq.~\ref{eqn:V_map_lift} yields $\Psi_0 = \Psi_1 A$.
Therefore, $\Range(\Psi_0) = \Range(\Psi_1)$, proving that $f$ is injective.

Next, we show that $f(\calW_{n-s+r,r}) \subset \calV_{FQS, E_{n,s}}$.
Consider an element $\Range(\Psi) \in \calW_{n-s+r,r}$ with $\Psi\in \R_*^{n-s+r\times r}$.
Since $\row_{1:r}(\Psi)$ is invertible, we have
\begin{equation}
    \Range(\Psi) = \Range(\Psi  \row_{1:r}(\Psi)^{-1} \row_{1:r}(M)).
\end{equation}
Using this, we obtain
\begin{equation}
    f(\Range(\Psi)) = 
    \Range\left(
    \begin{bmatrix}
        \row_{1:r}(M) \\
        \row_{r+1:s}(M) \\
        \row_{r+1:n-s+r}(\Psi) \row_{1:r}(\Psi)^{-1} \row_{1:r}(M)
    \end{bmatrix}
    \right),
\end{equation}
which is obviously an element of $\calV_{FQS, E_{n,s}}$ thanks to Eq.~\ref{eqn:V_reduced}.

Finally, to show that $f$ is a diffeomorphism of $\calW_{n-s+r,r}$ onto $\calV_{FQS, E_{n,s}}$, we construct a smooth inverse map.
The inverse is the restriction to $\calV_{FQS, E_{n,s}} \subset \calW_{n,r}$ (recall Eq.~\ref{eqn:VcontW}) of the smooth map $g:\calW_{n,r} \to \calW_{n-s+r,r}$ defined by
\begin{equation}
    g:
    \underbrace{\Range\left(
    \begin{bmatrix}
        I_r \\
        \row_{1:s-r}(A) \\
        \row_{s-r+1:n-r}(A)
    \end{bmatrix}
    \right)}_{\phi_{n,r}(A)} \mapsto
    \underbrace{\Range\left(
    \begin{bmatrix}
        I_r \\
        \row_{s-r+1:n-r}(A)
    \end{bmatrix}
    \right)}_{\phi_{n-s+r,r}(\row_{s-r+1:n-r}(A))}.
\end{equation}
This map is well-defined and smooth thanks to Corollary~\ref{cor:Wqp}.

To show that $f \circ g$ is the identity on $\calV_{FQS, E_{n,s}}$, let $\Range(\Phi) \in \calV_{FQS, E_{n,s}}$ for an $n\times r$ matrix $\Phi$.
Since $\Range(\row_{1:s}(\Phi)) = \Range(M)$ and $\row_{1:r}(M)$ is invertible, it follows that $\row_{1:r}(\Phi)$ is inverible.
By replacing $\Phi$ with $\Phi \row_{1:r}(\Phi)^{-1}$, we may assume without loss of generality that $\row_{1:r}(\Phi) = I_r$.
With this choice for $\Phi$, we observe that $\row_{1:s}(M) \row_{1:r}(M)^{-1} = \row_{1:s}(\Phi)$.
Using this observation and the definitions of $f$ and $g$, we calculate
\begin{equation}
\begin{aligned}
    f \circ g \big(\Range(\Phi)\big) 
    &= f \circ g
    \left( \Range\left( 
    \begin{bmatrix}
        I_r \\
        \row_{r+1:s}(\Phi) \\
        \row_{s+1:n}(\Phi)
    \end{bmatrix}
    \right) \right) \\
    &= f \left( \Range\left( 
    \begin{bmatrix}
        I_r \\
        \row_{s+1:n}(\Phi)
    \end{bmatrix}
    \right) \right) \\
    &= \Range \left( \bar{f} \left( 
    \begin{bmatrix}
        I_r \\
        \row_{s+1:n}(\Phi)
    \end{bmatrix}
    \right) \right) \\
    &= \Range \left( 
    \begin{bmatrix}
        I_r \\
        \row_{r+1:s}(M)\row_{1:r}(M)^{-1} \\
        \row_{s+1:n}(\Phi)
    \end{bmatrix}
    \right) \\
    &= \Range(\Phi).
\end{aligned}
\end{equation}
This proves that $f \circ g$ is the identity on $\calV_{FQS, E_{n,s}}$.

To prove that $g \circ f$ is the identity on $\calW_{n-s+r}$, we recall that $f(\calW_{n-s+r,r}) \subset \calV_{FQS, E_{n,s}}$.
Combining this with the result that $f \circ g$ is the identity on $\calV_{FQS, E_{n,s}}$ yields
\begin{equation}
    f \circ g \circ f = f.
\end{equation}
Since $f$ is injective, it follows that $g\circ f$ is the identity on $\calW_{n-s+r,r}$, proving that $\calV_{FQS, E_{n,s}}$ is diffeomorphic to $\calW_{n-s+r,r}$, and hence to $\R^{(n-s)\times r}$.

To estimate the diameter of $\calV_{F,X}$ in the Grassmannian $\calG_{n,r}$, it suffices to study $\calV_{FQS, E_{n,s}}$ since $\calV_{F,X}$ is related to $\calV_{FQS, E_{n,s}}$ by a unitary transformation.
The diffeomorphism $h=f\circ \phi_{n-s+r,r}:\R^{(n-s)\times r} \to \calV_{FQS, E_{n,s}}$ constructed above can be written explicitly as
\begin{equation}
    h: A \mapsto
    \Range\left(
    \begin{bmatrix}
        I_r \\
        \row_{r+1:s}(M)\row_{1:r}(M)^{-1} \\
        A
    \end{bmatrix}
    \right).
\end{equation}
Letting the columns of $U\in \R^{s\times r}$ be an orthonormal basis for $$\Range\left(
    \begin{bmatrix}
        I_r \\
        \row_{r+1:s}(M)\row_{1:r}(M)^{-1}
    \end{bmatrix}
\right),$$
a change of coordinates on $\R^{(n-s)\times r}$ yields a new diffeomorphism $\tilde{h}=:\R^{(n-s)\times r} \to \calV_{FQS, E_{n,s}}$ given by
\begin{equation}
    \tilde{h}: A \mapsto
    \Range\left(
    \begin{bmatrix}
        U \\
        A
    \end{bmatrix}
    \right).
\end{equation}
Let $p = \min\{ n-s, r \}$ and let $U_0$ denote the first $p$ columns of $U$ and let $U_1$ denote the remaining $r-p$ columns of $U$.
For each $\varepsilon \in (0,1]$, we define a subspace
\begin{equation}
    V_{\varepsilon} = 
    \tilde{h}\left(\begin{bmatrix}
        \frac{\sqrt{1-\varepsilon^2}}{\varepsilon} I_p & 0 \\
        0 & 0
    \end{bmatrix}\right)
    = \Range\left(
    \begin{bmatrix}
        \varepsilon U_0 & U_1 \\
        \sqrt{1-\varepsilon^2} I_p & 0_{p\times(r-p)}\\
        0_{(n-s-p)\times p} & 0_{(n-s-p)\times (r-p)}
    \end{bmatrix}
    \right),
\end{equation}
where we have indicated dimensions of zero entries for clarity.
We observe that the matrix on the right, which we call $\Phi_{\varepsilon}$, has orthonormal columns forming a basis for $V_{\varepsilon}$.
Since we have
\begin{equation}
    \Phi_1^T \Phi_{\varepsilon}
    = \begin{bmatrix}
        \varepsilon I_{p} & 0 \\
        0 & I_{r-p}
    \end{bmatrix},
\end{equation}
the principal angles $\theta_i(V_1, V_{\varepsilon})$ between these subspaces (see Bj\"{o}rck and Golub \cite{Bjorck1973numerical}) satisfy
\begin{equation}
    \cos\big(\theta_i(V_1, V_{\varepsilon}) \big)
    = \left\{
    \begin{matrix}
        1, & 1 \leq i \leq r-p \\
        \varepsilon, & r-p+1 \leq i \leq r.
    \end{matrix}
    \right.
\end{equation}
Therefore, the geodesic distance between the subspaces (see Bendokat et al. \cite{Bendokat2020grassmann}) approaches
\begin{equation}
    d(V_1, V_{\varepsilon}) = \sqrt{\sum_{i=1}^r \theta_i(V_1, V_{\varepsilon})^2 }
    \ \to \ \frac{\pi}{2} \sqrt{p}
\end{equation}
as $\varepsilon \to 0$.
This completes the proof of the theorem.

\section{Proof of Theorem~\ref{thm:gradient_lifting}}
\label{app:gradient_lifting}

\begin{proof}[Proof of \cref{lem:feature_map_differentiability}]
In the particular case when $K\in C^{2}(\bar{\calX}\times \bar{\calX})$, it is shown in step~2 of Theorem~1 in Zhou \cite{Zhou2008derivative} ($\alpha = 0$ case) that the difference quotients converge
\begin{equation}
    \frac{1}{t} \left( K_{x+t e_j} - K_{x} \right) \to (\partial^{e_j} K)_x 
    \quad \mbox{in $\calH$}\quad \mbox{as $t\to 0$}
    \label{eqn:convergence_of_kernel_difference_quotients}
\end{equation}
for every $x\in \calX$.
Using these partial derivatives, we establish that $\Phi_K$ is Fr\'{e}chet differentiable by employing Proposition~1.1.4 in Kesavan \cite{Kesavan2022nonlinear}.
To apply this result, it remains to show that each of the maps $x\mapsto (\partial^{e_j} K)_x$ are continuous.
By \cref{eqn:derivative_reproducing_property} (see Theorem~1 in \cite{Zhou2008derivative}) and symmetry of the reproducing kernel, we have
\begin{equation}
    \left\langle (\partial^{e_i} K)_x, \ (\partial^{e_j} K)_y \right\rangle 
    = \partial^{e_i} (\partial^{e_j} K)_y(x)
    = \partial^{(e_j, e_i)}K(y,x)
    = \partial^{(e_i, e_j)}K(x,y),
    \label{eqn:derivative_kernel_inner_product}
\end{equation}
which is a continuous function of $x,y\in\calX$.
Therefore,
\begin{equation}
    \left\Vert (\partial^{e_j} K)_x - (\partial^{e_j} K)_y \right\Vert^2
    = \partial^{(e_j, e_j)}K(x,x) - 2 \partial^{(e_j, e_j)}K(x,y) + \partial^{(e_j, e_j)}K(y,y),
\end{equation}
which approaches $0$ as $y\to x$ in $\calX$.
This proves that $x\mapsto (\partial^{e_j} K)_x$ is continuous on $\calX$ and so $\Phi_K$ is Fr\'{e}chet differentiable on $\calX$.

Since the feature map is Fr\'{e}chet differentiable, the first expression in \cref{eqn:derivative_of_feature_map} follows from \cref{eqn:convergence_of_kernel_difference_quotients} and Proposition~1.1.3 in \cite{Kesavan2022nonlinear}.
Given any $v\in\R^n$ and $f\in\calH$, \cref{eqn:derivative_reproducing_property} can be used to show that
\begin{equation}
    \left\langle \D\Phi_K(x) v, \ f \right\rangle
    = \sum_{j=1}^n v_j \left\langle \partial^{e_j} K)_x ,\ f \right\rangle
    = \sum_{j=1}^n v_j \partial^{e_j} f(x)
    = \left\langle v, \ \grad f(x) \right\rangle,
\end{equation}
which proves the second first expression in \cref{eqn:derivative_of_feature_map}.

Using \cref{eqn:derivative_of_feature_map} and \cref{eqn:derivative_kernel_inner_product}, we observe that 
\begin{equation}
    H(x,y) := \D\Phi_K(x)^* \D\Phi_K(y) =
    \begin{bmatrix}
        \partial^{(e_1, e_1)}K(x,y) & \cdots & \partial^{(e_1, e_n)}K(x,y) \\
        \vdots & \ddots & \vdots \\
        \partial^{(e_n, e_1)}K(x,y) & \cdots & \partial^{(e_n, e_n)}K(x,y)
    \end{bmatrix}
\end{equation}
is a continuous matrix-valued function of $x,y\in\calX$.
Therefore,
\begin{equation}
\begin{aligned}
    \left\Vert \D\Phi_K(x) - \D\Phi_K(y) \right\Vert_{\text{op}}^2
    &= \sup_{\substack{v\in\R^n : \\ \Vert v \Vert \leq 1}} \left\Vert \D\Phi_K(x)v - \D\Phi_K(y)v \right\Vert^2 \\
    &= \sup_{\substack{v\in\R^n : \\ \Vert v \Vert \leq 1}} v^T \big( H(x,x) - H(x,y) - H(y,x) + H(y,y) \big) v \\
    &\leq \left\Vert H(x,x) - H(x,y) - H(y,x) + H(y,y) \right\Vert_{\text{op}}
\end{aligned}
\end{equation}
which approaches $0$ as $y\to x$ in $\calX$, proving that $x\mapsto \D\Phi_K(x)$ is continuous with respect to the operator norm.
\end{proof}

\begin{proof}[Proof of \cref{thm:gradient_lifting}]
Since the kernel is continuous on $\bar{\calX}$, so is the feature map.
Because the feature map is continuous and injective on the compact set $\bar{\calX}$, Corollary~13.27 in Sutherland \cite{Sutherland2009introduction} shows that $\Phi_K$ is a topological embedding of $\bar{\calX}$ into $\calH$.
Since the restriction of a topological embedding is still an embedding, $\Phi_K$ is a topological embedding of $\calX$.
Let $\calX'$ denote the image of $\calX$ under $\Phi_K$ and let $\Psi_K: \calX' \to \calX$ denote the continuous inverse of $\left.\Phi_K\right\vert_{\calX}$ on its image.
To show that $\Phi_K$ is a $C^1$ diffeomorphism of $\calX$ onto $\calX'$, it remains to show that $\Psi_K$ is $C^1$.

Choose $x\in\calX$.
The expression \cref{eqn:kernel_derivative_Gram_matrix} for the derivative Gram matrix $G(x) = \D\Phi_K(x)^*\D\Phi_K(x)$ in terms of the kernel is verified directly using \cref{eqn:derivative_of_feature_map} and symmetry of the kernel.
Since $G(x)$ is positive-definite, $\D\Phi_K(x)$ is injective.
Moreover, the subspace $\Range\big( \D \Phi_K (x) \big)$ is finite-dimensional, hence it is closed in $\calH$.
Consequently, Proposition~3.2.8 in Margalef-Roig and Dominguez \cite{Margalef1992differential} shows that $\Phi_K$ is an immersion at $x$.
By Proposition~3.2.13 in \cite{Margalef1992differential} there is an open neighborhood $\calU$ of $x$ in $\calX$ so that $\Phi_K(\calU)$ is a $C^1$ submanifold of $\calH$ and $\Phi_K: \calU \to \Phi_K(\calU)$ is a $C^1$ diffeomorphism.
Because $\Phi_K$ is a topological embedding, $\Phi_K(\calU)$ is an open subset of $\calX'$ and the inverse of $\left.\Phi_K\right\vert_{\calU}$ on $\Phi_K(\calU)$ agrees with $\Psi_K$.
Therefore, $\Psi_K$ is $C^1$ on each $\Phi_K(\calU)$.
Since these sets cover $\calX'$, we have shown that $\Psi_K$ is $C^1$, and so $\Psi_K:\calX \to \calX'$ is a $C^1$ diffeomorphism.

Consequently, $F' = F \circ \Psi_K$ is well-defined and Fr\'{e}chet differentiable on $\calX'$.
Let $\phi$ be any real-valued differentiable function on $\R^m$.
Differentiating the relation $\phi \circ F = \phi \circ  F' \circ \Phi_K$ at $x\in\calX$ gives
\begin{equation}
    \grad ( \phi \circ F )(x) = \D \Phi_K(x)^* \grad ( \phi \circ F' )(K_x).
\end{equation}
Since $\grad ( \phi \circ F' )(K_x)$ is an element of $T_{K_x} \calX' = \Range\big( \D \Phi_K(x) \big)$, there is a vector $w\in\R^n$ so that
\begin{equation}
    \grad ( \phi \circ F' )(K_x) = \D \Phi_K(x) w.
\end{equation}
Since the derivative Gram matrix $G(x)$ is invertible, the above expressions can be solved for $w$, yielding \cref{eqn:lifted_gradient}.
\end{proof}

\section{Kernel-based model parameters}
\label{app:kernel_params}

\begin{table}[h!]
    \centering
    \caption{Kernel ridge regression parameters used to fit the dynamics and reconstruction maps for the jet flow in the nonlinear feature spaces identified by K-CoBRAS and KPCA.}
    \begin{tabular}{|c|c|c|c|c|}
        \hline
        coordinates & dimension $r$ & mapping & KRR regularization $\alpha$  & RBF width $\gamma$ \\
        \hline
        K-CoBRAS & $15$ & $\tilde{g}$ & $0.1000$ & $0.1000$ \\
        K-CoBRAS & $15$ & $\tilde{f}$ & $7.943\times 10^{-6}$ & $5.012\times 10^{-4}$ \\
        \hline
        K-CoBRAS & $30$ & $\tilde{g}$ & $0.005012$ & $0.01585$ \\
        K-CoBRAS & $30$ & $\tilde{f}$ & $3.981\times 10^{-5}$ & $0.001$ \\
        \hline
        KPCA & $15$ & $\tilde{g}$ & $0.02512$ & $0.03981$ \\
        KPCA & $15$ & $\tilde{f}$ & $1.585\times 10^{-4}$ & $2.512\times 10^{-3}$ \\
        \hline
        KPCA & $30$ & $\tilde{g}$ & $0.01259$ & $0.01585$ \\
        KPCA & $30$ & $\tilde{f}$ & $6.310\times 10^{-5}$ & $5.012\times 10^{-4}$ \\
        \hline
    \end{tabular}
    \label{tab:KRR_params}
\end{table}

\end{document}